\documentclass[onecolumn, conference]{IEEEtran}

\usepackage{amssymb}
\usepackage{amsmath}
\usepackage{epic,eepic}
\usepackage{epsfig}
\usepackage{cite}
\usepackage{verbatim}
\usepackage{enumerate}
\usepackage{subfigure}
\usepackage{multicol}
\usepackage{picinpar}
\usepackage{pstricks}
\usepackage{pst-node}

\def\QED{\mbox{\rule[0pt]{1.5ex}{1.5ex}}}

\newcommand{\define}{\stackrel{\triangle}{=}}

\newtheorem{theorem}{\bf Theorem}

\newtheorem{lemma}{\bf Lemma}

\newcommand{\xH}{\mathbf{H}}
\newcommand{\xX}{\mathbf{X}}
\newcommand{\xY}{\mathbf{Y}}
\newcommand{\xZ}{\mathbf{Z}}

\newcommand{\xw}{\mathbf{w}}
\newcommand{\xV}{\mathbf{V}}
\newcommand{\xI}{\mathbf{I}}

\begin{document}

\setcounter{page}{1}
\title{ Capacity of Wireless Networks within $o(\log(\mbox{SNR}))$ - the Impact of Relays, Feedback, Cooperation and Full-Duplex Operation}
\author{\authorblockN{Viveck R. Cadambe, Syed A. Jafar}
\authorblockA{Electrical Engineering and Computer Science\\
University of California Irvine, \\
Irvine, California, 92697, USA\\
Email: {vcadambe@uci.edu, syed@uci.edu}\\ \vspace{-1cm}}}

\maketitle
\vspace{10pt}
\begin{abstract} 
Recent work has characterized the sum capacity of time-varying/frequency-selective wireless interference networks and $X$ networks within $o(\log(\mbox{SNR}))$, i.e., with an accuracy approaching 100\% at high SNR (signal to noise power ratio). In this paper, we seek similar capacity characterizations for  wireless networks with relays, feedback, full duplex operation, and transmitter/receiver cooperation through noisy channels. First, we consider a network with $S$ source nodes, $R$ relay nodes and $D$ destination nodes with random time-varying/frequency-selective channel coefficients and global channel knowledge at all nodes. We allow full-duplex operation at all nodes, as well as causal noise-free feedback of all received signals to all source and relay nodes. The sum capacity of this network is characterized as $\frac{SD}{S+D-1}\log(\mbox{SNR})+o(\log(\mbox{SNR}))$. The implication of the result is that the capacity benefits of relays, causal feedback, transmitter/receiver cooperation through physical channels and full duplex operation become a negligible fraction of the network capacity at high SNR. Some exceptions to this result are also pointed out in the paper. Second, we consider a network with $K$ full duplex nodes with an independent message from every node to every other node in the network. We find that the sum capacity of this network is bounded below by $\frac{K(K-1)}{2K-2}+o(\log(\mbox{SNR}))$ and bounded above by $\frac{K(K-1)}{2K-3}+o(\log(\mbox{SNR}))$.
\end{abstract}

\newpage
\section{Introduction}
The recent surge of interest in approximate capacity characterizations of wireless networks has lead to substantial progress on several long standing open problems. The capacity of the $2$ user interference channel and the capacity of certain relay networks have been characterized within a constant number of bits \cite{tse_etkin_wang:int1bit,diggavi_tse:deterministic_relay}. The sum capacity of the $K$-user time-varying/frequency-selective interference channel (see Figure \ref{fig:dofintx}) was characterized in \cite{cadambe_jafar:Kuserint} as 
\begin{equation} \label{eqn:dofint}C(\mbox{SNR}) = \frac{K}{2} \log(\mbox{SNR}) + o(\log(\mbox{SNR})) \end{equation}
where $\mbox{SNR}$ represents the signal to noise ratio (equivalently, the total transmit power of all nodes when the local noise power at each receiver is normalized to unity).
If the sum capacity of a network is characterized as $C(\mbox{SNR}) = d \log(\mbox{SNR}) + o(\log(\mbox{SNR}))$, then we say that the network has $d$ degrees of freedom (also known as multiplexing gain or capacity pre-log).
 Since, by definition, at high SNR the $o(\log(\mbox{SNR}))$ term is a vanishing fraction of $\log(\mbox{SNR})$, the accuracy of such a capacity characterization approaches $100\%$ as the SNR approaches infinity.  The result that the $K$ user interference network has $K/2$ degrees of freedom is interesting because it shows that at high SNR, every user in an interference network is simultaneously able to reliably communicate at a rate close to half of his individual capacity in the absence of all interferers.  The result reveals the fallacy of the conventional ``cake-cutting" view of orthogonal medium access because, essentially, it implies that everyone gets "half the cake" (where the cake represents the rate achievable by the user in the absence of interference). The achievable scheme of \cite{cadambe_jafar:Kuserint} is based on the idea of interference alignment (see \cite{jafar_shamai:dofx} and references therein). The key to interference alignment is the realization that the alignment of signal dimensions (in time, frequency, space and codes) is relative to the observer (receiver). Since every receiver sees a different picture, signals may be constructed to cast overlapping shadows at the receivers where they constitute interference while they remain distinguishable at the receivers where they are desired. 

$X$ networks are a generalization of interference networks. Unlike an interference network where each transmitter has a message for only its corresponding receiver, in an $X$ network every transmitter has an independent message for every receiver. Reference \cite{cadambe_jafar:dofx} studied the $S\times D$ $X$ network (Figure \ref{fig:dofintx}), i.e., a network with $S$ transmitters, $D$ receivers and $SD$ independent messages - one message for each transmitter-receiver pair. Using an interference alignment based achievable scheme, \cite{cadambe_jafar:dofx} characterized the sum capacity of the $S \times D$ $X$ network as
\begin{equation}
\label{eqn:dofx}
C(\mbox{SNR})=\frac{SD}{S+D-1}\log(\mbox{SNR})+o(\log(\mbox{SNR})).
\end{equation}
In other words, the $S \times D$ $X$ network has $\frac{SD}{S+D-1}$ degrees of freedom. For $S=D=K$ it is interesting to note that the degrees of freedom benefits of $X$ networks over interference networks diminish as the number of users, $K$, increases.

The results of \cite{cadambe_jafar:dofx} and \cite{cadambe_jafar:Kuserint} show that interference alignment suffices to achieve the capacity of wireless interference and $X$ networks at high SNR. However, many interesting possibilities are not explored by these scenarios. In general, a wireless network can have relay nodes. The nodes may be capable of full-duplex operation so that they can all transmit and receive simultaneously and thus co-operate to improve rates of communication. Furthermore, even when perfect global channel knowledge is already assumed, noise-free (causal) feedback of received signals may have significant benefits in a wireless network. For $K$-user interference networks, degrees of freedom with relays and transmitter/receiver cooperation through noisy channels have been explored in \cite{host-madsen_nosratinia:dofint} and an outerbound of $K/2$ is obtained. Since the results of \cite{cadambe_jafar:Kuserint} show that $K/2$ degrees of freedom can be achieved even without relays and cooperation, the conclusion is that relays and transmitter/receiver cooperation cannot increase the degrees of freedom of interference networks. In this paper, we seek a generalization of the results of \cite{host-madsen_nosratinia:dofint} from interference networks to $X$ networks. In particular, we seek a degrees of freedom characterization for fully connected wireless $X$ networks with relays, feedback, transmitter (receiver) cooperation and full duplex operation. 

\begin{figure}[!tbp]
\begin{center}\setlength{\unitlength}{0.00041667in}
\begingroup\makeatletter\ifx\SetFigFont\undefined%
\gdef\SetFigFont#1#2#3#4#5{%
  \reset@font\fontsize{#1}{#2pt}%
  \fontfamily{#3}\fontseries{#4}\fontshape{#5}%
  \selectfont}%
\fi\endgroup%
{\renewcommand{\dashlinestretch}{30}
\begin{picture}(16934,6385)(0,-10)
\put(5850,4033){\makebox(0,0)[lb]{\smash{{{\SetFigFont{6}{7.2}{\rmdefault}{\mddefault}{\updefault}$\widehat{W}_{2}$}}}}}
\put(5850,1933){\makebox(0,0)[lb]{\smash{{{\SetFigFont{6}{7.2}{\rmdefault}{\mddefault}{\updefault}$\widehat{W}_{K}$}}}}}
\put(5925,5533){\makebox(0,0)[lb]{\smash{{{\SetFigFont{6}{7.2}{\rmdefault}{\mddefault}{\updefault}$\widehat{W}_{1}$}}}}}
\put(1275,4033){\makebox(0,0)[lb]{\smash{{{\SetFigFont{6}{7.2}{\rmdefault}{\mddefault}{\updefault}$X_{2}$}}}}}
\put(1275,5533){\makebox(0,0)[lb]{\smash{{{\SetFigFont{6}{7.2}{\rmdefault}{\mddefault}{\updefault}$X_{1}$}}}}}
\put(0,4033){\makebox(0,0)[lb]{\smash{{{\SetFigFont{6}{7.2}{\rmdefault}{\mddefault}{\updefault}$W_{2}$}}}}}
\put(0,5533){\makebox(0,0)[lb]{\smash{{{\SetFigFont{6}{7.2}{\rmdefault}{\mddefault}{\updefault}$W_{1}$}}}}}
\put(15150,3733){\makebox(0,0)[lb]{\smash{{{\SetFigFont{6}{7.2}{\familydefault}{\mddefault}{\updefault}$\widehat{W}_{2,S}$}}}}}
\put(10500,5683){\makebox(0,0)[lb]{\smash{{{\SetFigFont{6}{7.2}{\rmdefault}{\mddefault}{\updefault}$X_1$}}}}}
\put(10500,4183){\makebox(0,0)[lb]{\smash{{{\SetFigFont{6}{7.2}{\rmdefault}{\mddefault}{\updefault}$X_2$}}}}}
\put(10425,2383){\makebox(0,0)[lb]{\smash{{{\SetFigFont{6}{7.2}{\rmdefault}{\mddefault}{\updefault}$X_S$}}}}}
\put(14100,1633){\makebox(0,0)[lb]{\smash{{{\SetFigFont{6}{7.2}{\rmdefault}{\mddefault}{\updefault}$Y_D$}}}}}
\put(15300,4108){\makebox(0,0)[lb]{\smash{{{\SetFigFont{6}{7.2}{\familydefault}{\mddefault}{\updefault}$~\vdots$}}}}}
\put(14100,4183){\makebox(0,0)[lb]{\smash{{{\SetFigFont{6}{7.2}{\rmdefault}{\mddefault}{\updefault}$Y_2$}}}}}
\put(14100,5683){\makebox(0,0)[lb]{\smash{{{\SetFigFont{6}{7.2}{\rmdefault}{\mddefault}{\updefault}$Y_1$}}}}}
\put(4950,1933){\makebox(0,0)[lb]{\smash{{{\SetFigFont{6}{7.2}{\rmdefault}{\mddefault}{\updefault}$Y_{K}$}}}}}
\put(4950,5533){\makebox(0,0)[lb]{\smash{{{\SetFigFont{6}{7.2}{\rmdefault}{\mddefault}{\updefault}$Y_{1}$}}}}}
\path(8625,4783)(9750,4783)(9750,3583)
	(8625,3583)(8625,4783)
\path(8625,6358)(9750,6358)(9750,5158)
	(8625,5158)(8625,6358)
\path(8550,2983)(9675,2983)(9675,1783)
	(8550,1783)(8550,2983)
\path(14925,2383)(16050,2383)(16050,1108)
	(14925,1108)(14925,2383)
\path(11100,2533)(13875,5683)
\path(13818.187,5573.126)(13875.000,5683.000)(13773.166,5612.788)
\put(4950,4033){\makebox(0,0)[lb]{\smash{{{\SetFigFont{6}{7.2}{\rmdefault}{\mddefault}{\updefault}$Y_{2}$}}}}}
\put(9075,4108){\makebox(0,0)[lb]{\smash{{{\SetFigFont{6}{7.2}{\familydefault}{\mddefault}{\updefault}$~\vdots$}}}}}
\put(9075,5683){\makebox(0,0)[lb]{\smash{{{\SetFigFont{6}{7.2}{\familydefault}{\mddefault}{\updefault}$~\vdots$}}}}}
\put(15300,1633){\makebox(0,0)[lb]{\smash{{{\SetFigFont{6}{7.2}{\familydefault}{\mddefault}{\updefault}$~\vdots$}}}}}
\put(15300,5608){\makebox(0,0)[lb]{\smash{{{\SetFigFont{6}{7.2}{\familydefault}{\mddefault}{\updefault}$~\vdots$}}}}}
\put(8925,2308){\makebox(0,0)[lb]{\smash{{{\SetFigFont{6}{7.2}{\familydefault}{\mddefault}{\updefault}$~\vdots$}}}}}
\put(15075,2008){\makebox(0,0)[lb]{\smash{{{\SetFigFont{6}{7.2}{\familydefault}{\mddefault}{\updefault}$\widehat{W}_{D,1}$}}}}}
\put(15075,1258){\makebox(0,0)[lb]{\smash{{{\SetFigFont{6}{7.2}{\familydefault}{\mddefault}{\updefault}$\widehat{W}_{D,S}$}}}}}
\put(15075,4483){\makebox(0,0)[lb]{\smash{{{\SetFigFont{6}{7.2}{\familydefault}{\mddefault}{\updefault}$\widehat{W}_{2,1}$}}}}}
\put(15075,5983){\makebox(0,0)[lb]{\smash{{{\SetFigFont{6}{7.2}{\familydefault}{\mddefault}{\updefault}$\widehat{W}_{1,1}$}}}}}
\put(15075,5233){\makebox(0,0)[lb]{\smash{{{\SetFigFont{6}{7.2}{\familydefault}{\mddefault}{\updefault}$\widehat{W}_{1,S}$}}}}}
\put(8700,2683){\makebox(0,0)[lb]{\smash{{{\SetFigFont{6}{7.2}{\familydefault}{\mddefault}{\updefault}$W_{1,S}$}}}}}
\put(8700,1933){\makebox(0,0)[lb]{\smash{{{\SetFigFont{6}{7.2}{\familydefault}{\mddefault}{\updefault}$W_{D,S}$}}}}}
\put(8850,4483){\makebox(0,0)[lb]{\smash{{{\SetFigFont{6}{7.2}{\familydefault}{\mddefault}{\updefault}$W_{1,2}$}}}}}
\put(8850,3733){\makebox(0,0)[lb]{\smash{{{\SetFigFont{6}{7.2}{\familydefault}{\mddefault}{\updefault}$W_{N,2}$}}}}}
\put(8850,6058){\makebox(0,0)[lb]{\smash{{{\SetFigFont{6}{7.2}{\familydefault}{\mddefault}{\updefault}$W_{1,1}$}}}}}
\put(8850,5308){\makebox(0,0)[lb]{\smash{{{\SetFigFont{6}{7.2}{\familydefault}{\mddefault}{\updefault}$W_{D,1}$}}}}}
\put(0,1933){\makebox(0,0)[lb]{\smash{{{\SetFigFont{6}{7.2}{\rmdefault}{\mddefault}{\updefault}$W_{K}$}}}}}
\put(1425,58){\makebox(0,0)[lb]{\smash{{{\SetFigFont{6}{7.2}{\rmdefault}{\mddefault}{\updefault}$C(\mbox{SNR}) = \frac{K}{2} \log(\mbox{SNR}) + o(\log(\mbox{SNR}))$ }}}}}
\put(10350,58){\makebox(0,0)[lb]{\smash{{{\SetFigFont{6}{7.2}{\rmdefault}{\mddefault}{\updefault}$C(\mbox{SNR}) = \frac{SD}{S+D-1} \log(\mbox{SNR}) + o(\log(\mbox{SNR}))$ }}}}}
\put(1800,658){\makebox(0,0)[lb]{\smash{{{\SetFigFont{7}{8.4}{\rmdefault}{\mddefault}{\updefault}The $K$ user interference channel}}}}}
\put(1200,1933){\makebox(0,0)[lb]{\smash{{{\SetFigFont{6}{7.2}{\rmdefault}{\mddefault}{\updefault}$X_{K}$}}}}}
\put(11175,658){\makebox(0,0)[lb]{\smash{{{\SetFigFont{7}{8.4}{\rmdefault}{\mddefault}{\updefault}The $S \times D$ $X$ channel}}}}}
\path(11100,2533)(13800,4183)
\path(13713.250,4094.828)(13800.000,4183.000)(13681.963,4146.024)
\path(1950,4108)(4650,5533)
\path(4557.877,5450.457)(4650.000,5533.000)(4529.871,5503.521)
\path(1950,2008)(4725,5533)
\path(4674.345,5420.155)(4725.000,5533.000)(4627.201,5457.268)
\path(1950,2008)(4650,4033)
\path(4572.000,3937.000)(4650.000,4033.000)(4536.000,3985.000)
\path(1950,5608)(4650,4183)
\path(4529.871,4212.479)(4650.000,4183.000)(4557.877,4265.543)
\path(1950,2008)(4650,2008)
\path(4530.000,1978.000)(4650.000,2008.000)(4530.000,2038.000)
\path(1950,4108)(4650,4108)
\path(4530.000,4078.000)(4650.000,4108.000)(4530.000,4138.000)
\put(10950,2489){\ellipse{212}{212}}
\put(13919,1708){\ellipse{212}{212}}
\put(10966,4279){\ellipse{212}{212}}
\put(10973,5798){\ellipse{212}{212}}
\put(13875,4289){\ellipse{212}{212}}
\put(13875,5789){\ellipse{212}{212}}
\put(1800,2008){\ellipse{212}{212}}
\put(1816,4129){\ellipse{212}{212}}
\put(1823,5648){\ellipse{212}{212}}
\put(4725,2008){\ellipse{212}{212}}
\put(4725,4139){\ellipse{212}{212}}
\put(4725,5639){\ellipse{212}{212}}
\path(14850,4783)(15975,4783)(15975,3583)
	(14850,3583)(14850,4783)
\path(14475,4223)(14850,4223)
\blacken\path(14726.000,4192.000)(14850.000,4223.000)(14726.000,4254.000)(14726.000,4192.000)
\path(14850,6283)(15975,6283)(15975,5083)
	(14850,5083)(14850,6283)
\path(14475,5723)(14850,5723)
\blacken\path(14726.000,5692.000)(14850.000,5723.000)(14726.000,5754.000)(14726.000,5692.000)
\path(1950,4108)(4650,2083)
\path(4536.000,2131.000)(4650.000,2083.000)(4572.000,2179.000)
\dashline{60.000}(13875,4033)(13875,2083)
\dashline{60.000}(10950,3958)(10950,2833)
\path(14550,1708)(14925,1708)
\blacken\path(14805.000,1678.000)(14925.000,1708.000)(14805.000,1738.000)(14805.000,1678.000)
\path(5400,4108)(5775,4108)
\blacken\path(5655.000,4078.000)(5775.000,4108.000)(5655.000,4138.000)(5655.000,4078.000)
\path(5475,5608)(5850,5608)
\blacken\path(5730.000,5578.000)(5850.000,5608.000)(5730.000,5638.000)(5730.000,5578.000)
\path(5400,2008)(5775,2008)
\blacken\path(5655.000,1978.000)(5775.000,2008.000)(5655.000,2038.000)(5655.000,1978.000)
\path(450,4108)(1125,4108)
\blacken\path(1005.000,4078.000)(1125.000,4108.000)(1005.000,4138.000)(1005.000,4078.000)
\path(450,5608)(1125,5608)
\blacken\path(1005.000,5578.000)(1125.000,5608.000)(1005.000,5638.000)(1005.000,5578.000)
\path(9675,2458)(10350,2458)
\blacken\path(10230.000,2428.000)(10350.000,2458.000)(10230.000,2488.000)(10230.000,2428.000)
\path(9750,4258)(10425,4258)
\blacken\path(10305.000,4228.000)(10425.000,4258.000)(10305.000,4288.000)(10305.000,4228.000)
\path(9750,5758)(10425,5758)
\blacken\path(10305.000,5728.000)(10425.000,5758.000)(10305.000,5788.000)(10305.000,5728.000)
\path(450,2008)(1125,2008)
\blacken\path(1005.000,1978.000)(1125.000,2008.000)(1005.000,2038.000)(1005.000,1978.000)
\path(11100,5758)(13875,1783)
\path(13781.710,1864.222)(13875.000,1783.000)(13830.908,1898.568)
\path(11100,4258)(13800,1708)
\path(13692.160,1768.584)(13800.000,1708.000)(13733.357,1812.205)
\path(11100,2533)(13800,1633)
\path(13676.671,1642.487)(13800.000,1633.000)(13695.645,1699.408)
\path(11100,5758)(13800,5758)
\path(13680.000,5728.000)(13800.000,5758.000)(13680.000,5788.000)
\path(11100,4258)(13800,5683)
\path(13707.877,5600.457)(13800.000,5683.000)(13679.871,5653.521)
\path(11100,5758)(13800,4333)
\path(13679.871,4362.479)(13800.000,4333.000)(13707.877,4415.543)
\path(11100,4258)(13800,4258)
\path(13680.000,4228.000)(13800.000,4258.000)(13680.000,4288.000)
\dashline{60.000}(1800,3808)(1800,2308)
\dashline{60.000}(4725,3883)(4725,2383)
\path(1950,5608)(4650,5608)
\path(4530.000,5578.000)(4650.000,5608.000)(4530.000,5638.000)
\path(1950,5608)(4725,2083)
\path(4627.201,2158.732)(4725.000,2083.000)(4674.345,2195.845)
\end{picture}
}\end{center}
\caption{Interference and $X$ channels}
\label{fig:dofintx}
\end{figure}

\section{Summary of Results}
\label{sec:discuss}

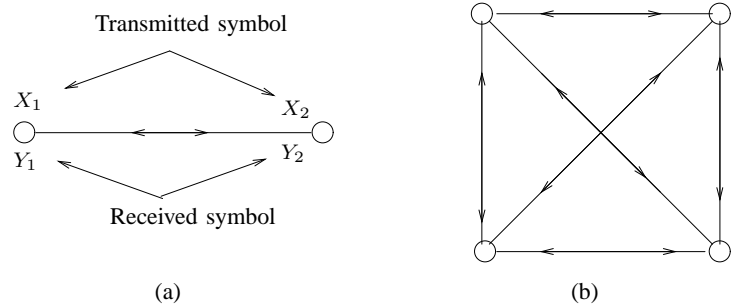
\begin{figure}[!tbp]
\begin{center}\setlength{\unitlength}{0.00052500in}
\begingroup\makeatletter\ifx\SetFigFont\undefined%
\gdef\SetFigFont#1#2#3#4#5{%
  \reset@font\fontsize{#1}{#2pt}%
  \fontfamily{#3}\fontseries{#4}\fontshape{#5}%
  \selectfont}%
\fi\endgroup%
{\renewcommand{\dashlinestretch}{30}
\begin{picture}(7247,3111)(0,-10)
\put(114,1783){\ellipse{212}{212}}
\put(3127,1783){\ellipse{212}{212}}
\put(4764,583){\ellipse{212}{212}}
\put(7133,583){\ellipse{212}{212}}
\put(4733,2983){\ellipse{212}{212}}
\put(7133,2983){\ellipse{212}{212}}
\path(5589.066,2169.360)(5483.000,2233.000)(5546.640,2126.934)
\path(5483,2233)(6383,1333)
\path(6276.934,1396.640)(6383.000,1333.000)(6319.360,1439.066)
\path(233,1783)(3008,1783)
\path(1328.000,1813.000)(1208.000,1783.000)(1328.000,1753.000)
\path(1208,1783)(1958,1783)
\path(1838.000,1753.000)(1958.000,1783.000)(1838.000,1813.000)
\path(1583,2608)(2633,2158)
\path(2510.885,2177.696)(2633.000,2158.000)(2534.520,2232.845)
\path(5333,2383)(6533,1183)
\path(5396.640,1289.066)(5333.000,1183.000)(5439.066,1246.640)
\path(5333,1183)(6533,2383)
\path(6469.360,2276.934)(6533.000,2383.000)(6426.934,2319.360)
\path(5453.000,3013.000)(5333.000,2983.000)(5453.000,2953.000)
\path(5333,2983)(6533,2983)
\path(6413.000,2953.000)(6533.000,2983.000)(6413.000,3013.000)
\path(4763.000,2263.000)(4733.000,2383.000)(4703.000,2263.000)
\path(4733,2383)(4733,883)
\path(4703.000,1003.000)(4733.000,883.000)(4763.000,1003.000)
\path(5453.000,613.000)(5333.000,583.000)(5453.000,553.000)
\path(5333,583)(6683,583)
\path(6563.000,553.000)(6683.000,583.000)(6563.000,613.000)
\path(7163.000,2413.000)(7133.000,2533.000)(7103.000,2413.000)
\path(7133,2533)(7133,1108)
\path(7103.000,1228.000)(7133.000,1108.000)(7163.000,1228.000)
\path(7133,2908)(7133,658)
\path(4733,2908)(4733,658)
\path(4883,583)(6983,583)
\path(7058,2908)(4808,658)
\path(4808,2908)(7058,658)
\path(4883,2983)(7058,2983)
\path(1511,1118)(461,1493)
\path(584.099,1480.892)(461.000,1493.000)(563.919,1424.387)
\path(1583,2608)(533,2233)
\path(635.919,2301.613)(533.000,2233.000)(656.099,2245.108)
\path(1496,1141)(2546,1516)
\path(2443.081,1447.387)(2546.000,1516.000)(2422.901,1503.892)
\put(983,808){\makebox(0,0)[lb]{{\SetFigFont{9}{10.8}{\rmdefault}{\mddefault}{\updefault}Received symbol}}}
\put(833,2758){\makebox(0,0)[lb]{{\SetFigFont{9}{10.8}{\rmdefault}{\mddefault}{\updefault}Transmitted symbol}}}
\put(8,1408){\makebox(0,0)[lb]{{\SetFigFont{8}{9.6}{\rmdefault}{\mddefault}{\updefault}$Y_1$}}}
\put(8,2008){\makebox(0,0)[lb]{{\SetFigFont{8}{9.6}{\rmdefault}{\mddefault}{\updefault}$X_1$}}}
\put(2708,1933){\makebox(0,0)[lb]{{\SetFigFont{8}{9.6}{\rmdefault}{\mddefault}{\updefault}$X_2$}}}
\put(2708,1483){\makebox(0,0)[lb]{{\SetFigFont{8}{9.6}{\rmdefault}{\mddefault}{\updefault}$Y_2$}}}
\put(5633,58){\makebox(0,0)[lb]{{\SetFigFont{9}{10.8}{\rmdefault}{\mddefault}{\updefault}(b)}}}
\put(1433,58){\makebox(0,0)[lb]{{\SetFigFont{9}{10.8}{\rmdefault}{\mddefault}{\updefault}(a)}}}
\end{picture}
}\end{center}
\caption{$K$ user full duplex network with (a) $K=2$, (b) $K=4$}
\label{fig:full_duplex}
\end{figure}
There are two main results in the paper. The first result applies to networks whose source and destination nodes are disjoint. In other words, if a node is a source node for a message, then it cannot be the destination node for any message in the network (and vice versa). In Section \ref{sec:SRDnet} we introduce the $S \times R \times D$ network - a fully connected network with $S$ source nodes, $R$ relay nodes and $D$ destination nodes with full duplex operation assumed at all nodes. Since nodes are full duplex, they can all transmit and receive and thus network allows co-operation between all nodes through noisy channels. Also included is perfect causal feedback channel of all received signals to all source and relay nodes. We show that this network has $\frac{SD}{S+D-1}$ degrees of freedom meaning that the capacity of the network can be expressed as 
$$ C(\mbox{SNR})=\frac{SD}{S+D-1}\log(\mbox{SNR})+o(\log(\mbox{SNR})) $$
Note that the above approximation is identical to the approximation in equation (\ref{eqn:dofx}). Achievability thus follows trivially from the interference alignment based achievable scheme of \cite{cadambe_jafar:dofx} between the $S$ source nodes and $D$ destination nodes. The outerbound is shown in Theorem \ref{thm:main} in Section \ref{subsec:main}. 

The second result of this  paper is for the $K$ user full duplex network considered in Section \ref{sec:fd}. There is an independent message from every node to each of the other $K-1$ nodes in the network so that there are a total of $K(K-1)$ messages. All nodes in the network function simultaneously as source, relay and destination nodes. In Theorem \ref{thm:Kuserfd}, the sum capacity of this network is bounded as 
\begin{eqnarray*} C(\mbox{SNR}) \geq \frac{K(K-1)}{2K-2} \log(\mbox{SNR}) + o(\log(\mbox{SNR})) \\ 
  C(\mbox{SNR}) \leq \frac{K(K-1)}{2K-3}\log(\mbox{SNR}) + o(\log(\mbox{SNR})) \end{eqnarray*}
The achievable scheme is based on interference alignment. Next, we discuss the implications of these results.

\subsection{Can Relays, Feedback, Co-operation and Full Duplex Operation Improve the Degrees of Freedom of a Wireless Network ?}
An interesting implication of the results is that the capacity improvements due to relays, feedback, full duplex operation and noisy co-operation become a negligible fraction of the network capacity as SNR increases. Put differently, these factors (relays etc.) do not improve the degrees of freedom of wireless networks. Thus, in most wireless networks, the search for capacity improvements of the order of $\log(\mbox{SNR})$ ends in interference alignment. While the result is obtained for a broad class of networks, it is limited by the underlying assumptions. For example, we assume that each node is equipped with only a single antenna, the channel coefficients are time-varying/frequency selective and drawn randomly from a continuous distribution whose support is bounded below by a non-zero constant (i.e., the network is fully connected), the channel knowledge is global and perfect. Some examples of scenarios where degrees of freedom can be improved are listed below.

\begin{enumerate}
\item Relays can improve the degrees of freedom if a network is not fully connected.
\item Co-operation can increase the degrees of freedom if the cost of co-operation is not accounted for (e.g., in genie-aided cognitive radio networks).
\item Full duplex operation can increase the degrees of freedom if the same node can be the source for one message and the destination for another message.
\item Feedback can improve the degrees of freedom if it is provided to a destination node, in which case it helps the decoder by providing, in effect, extra antennas that can be used to null out interference.
\end{enumerate}

In the remainder of this section, we discuss the above statements in detail.

\subsection{Relays}
\label{subsec:relay}
\begin{figure}[!tbp]
\begin{center}\setlength{\unitlength}{0.00040833in}
\begingroup\makeatletter\ifx\SetFigFont\undefined%
\gdef\SetFigFont#1#2#3#4#5{%
  \reset@font\fontsize{#1}{#2pt}%
  \fontfamily{#3}\fontseries{#4}\fontshape{#5}%
  \selectfont}%
\fi\endgroup%
{\renewcommand{\dashlinestretch}{30}
\begin{picture}(8996,5645)(0,-10)
\put(3825,5247){\ellipse{212}{212}}
\put(3825,4797){\ellipse{212}{212}}
\put(3794,1872){\ellipse{212}{212}}
\put(3794,1347){\ellipse{212}{212}}
\put(6719,4347){\ellipse{212}{212}}
\put(6750,2953){\ellipse{212}{212}}
\put(900,4347){\ellipse{212}{212}}
\put(900,2891){\ellipse{212}{212}}
\put(900,3522){\ellipse{600}{2400}}
\put(6745,3668){\ellipse{600}{2400}}
\put(3825,3147){\ellipse{750}{4950}}
\path(975,4347)(3750,5247)
\path(975,4347)(3750,4797)
\path(3636.349,4748.178)(3750.000,4797.000)(3626.745,4807.405)
\path(975,4347)(3750,1872)
\path(3640.476,1929.485)(3750.000,1872.000)(3680.413,1974.263)
\path(975,4347)(3750,1347)
\path(3646.492,1414.721)(3750.000,1347.000)(3690.538,1455.463)
\path(3825,1347)(6600,4347)
\path(6540.538,4238.537)(6600.000,4347.000)(6496.492,4279.279)
\path(3825,1347)(6600,2922)
\path(6510.446,2836.677)(6600.000,2922.000)(6480.830,2888.858)
\path(4800,297)(3900,1047)
\path(4011.392,993.225)(3900.000,1047.000)(3972.981,947.131)
\path(825,5172)(900,4797)
\path(847.049,4908.786)(900.000,4797.000)(905.883,4920.553)
\dashline{60.000}(6750,4197)(6750,3072)
\dashline{60.000}(3825,4647)(3825,2022)
\dashline{60.000}(900,4197)(900,3072)
\path(1050,2922)(3750,5247)
\path(3678.643,5145.964)(3750.000,5247.000)(3639.492,5191.430)
\path(1050,2847)(3750,1272)
\path(3631.230,1306.551)(3750.000,1272.000)(3661.463,1358.378)
\path(1050,2847)(3750,1797)
\path(3627.286,1812.533)(3750.000,1797.000)(3649.033,1868.454)
\path(975,2847)(3750,4722)
\path(3667.365,4629.959)(3750.000,4722.000)(3633.774,4679.675)
\path(3900,5322)(6675,4422)
\path(6551.598,4430.484)(6675.000,4422.000)(6570.108,4487.557)
\path(3900,5172)(6675,2997)
\path(3900,4797)(6675,2997)
\path(6557.999,3037.134)(6675.000,2997.000)(6590.650,3087.472)
\path(3900,4872)(6600,4347)
\path(6476.480,4340.456)(6600.000,4347.000)(6487.932,4399.353)
\path(3900,1872)(6675,4272)
\path(6603.861,4170.811)(6675.000,4272.000)(6564.612,4216.193)
\path(3900,1797)(6600,2922)
\path(6500.769,2848.154)(6600.000,2922.000)(6477.692,2903.538)
\path(6939,5253)(6864,4803)
\path(6854.136,4926.299)(6864.000,4803.000)(6913.320,4916.435)
\put(4125,72){\makebox(0,0)[lb]{{\SetFigFont{8}{9.6}{\familydefault}{\mddefault}{\updefault}$R$ distributed relays}}}
\put(6150,5397){\makebox(0,0)[lb]{{\SetFigFont{8}{9.6}{\familydefault}{\mddefault}{\updefault}$K$ distributed destinations}}}
\put(0,5247){\makebox(0,0)[lb]{{\SetFigFont{8}{9.6}{\familydefault}{\mddefault}{\updefault}$K$ distributed sources}}}
\end{picture}
}\end{center}
\caption{The parallel relay channel}
\label{fig:parallel_relay}
\end{figure}
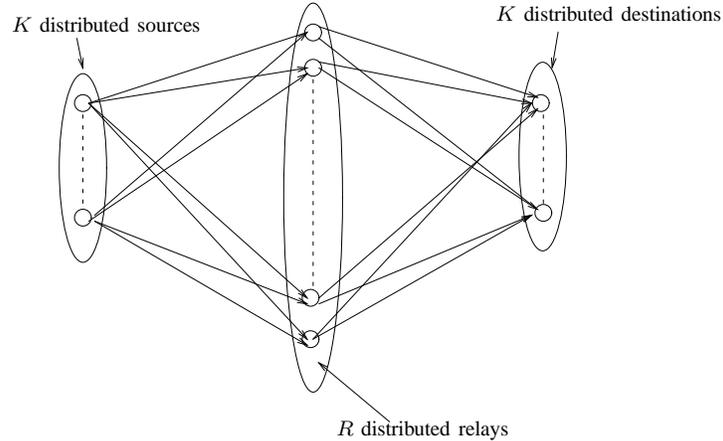
It is easy to construct examples of networks that are not fully connected where the presence of relays increases the degrees of freedom. A simple example would be a 3 node network where the channel coefficient between the source and destination is zero, so that a non-zero capacity can only be achieved through the relay node. We further illustrate the significance of this result by placing it into perspective with some existing results on relay networks.

The benefits of relays on the degrees of freedom of networks have been considered in \cite{borade_zheng_gallager:afrelay}, \cite{ozgur_paulraj_etal:dofmimorelay}. \cite{borade_zheng_gallager:afrelay} considered a single $K$-antenna transmitting node communicating with a $K$-antenna receiving node through multiple orthogonal hops of distributed parallel relays.  Using a simple amplify and forward scheme they showed that the capacity of this network scales as $K\log(\mbox{SNR})+o(\log(\mbox{SNR}))$. Reference \cite{ozgur_paulraj_etal:dofmimorelay} considered a $2$ hop parallel relay network with $K$ distributed sources and $K$ distributed destinations, with the sources and destinations separated by a layer of $R$ distributed relays (Figure \ref{fig:parallel_relay}). Like the interference network, the parallel relay network has $K$ messages, one from each transmitter to its unique corresponding receiver. In this setting, the results of \cite{ozgur_paulraj_etal:dofmimorelay} indicate that if the relays are full-duplex and the number of relays $R \to \infty$, the sum capacity approaches $K \log(\mbox{SNR})$ so that the network has $K$ degrees of freedom using an achievable scheme based on amplify and forward. Reference \cite{cadambe_jafar:dofx} shows a similar result for parallel relay networks by viewing each hop as an $X$ network. The achievable scheme of \cite{cadambe_jafar:dofx} uses decode-and-forward and full-duplex relays to achieve $\frac{KR}{(K+R-1)}$ degrees of freedom over this network. Observe that by increasing $R$, we can get arbitrarily close to $K$ degrees of freedom. Thus, if we assume that source nodes are not directly connected to the destination nodes, as in Figure \ref{fig:parallel_relay}, then the results of \cite{cadambe_jafar:dofx} and \cite{ozgur_paulraj_etal:dofmimorelay} provide interesting examples where full-duplex relays increase the degrees of freedom from $K/2$ to $K$. However, the key difference between the parallel relay networks in these cases and the model in this paper is that the former network is \emph{not fully connected} since the links from the source nodes to destination nodes are absent. If the parallel relay network of Figure \ref{fig:parallel_relay} is fully connected, i.e., if there are non-zero channel coefficients from all source nodes to all the destination nodes, surprisingly, the resulting network has only $K/2$ degrees of freedom  (by the result of Theorem \ref{thm:main}). In other words, the presence of direct links from source to destination nodes in this network reduces the degrees of freedom of the network by a factor of half. By artificially imposing the half-duplex constraint on the relays, the relay-aided schemes of \cite{cadambe_jafar:dofx} (and \cite{ozgur_paulraj_etal:dofmimorelay}) can provide only an alternate means of achieving $K/2$ degrees of freedom through a 2-phase operation.

\subsection{Cognitive Cooperation}
\label{subsec:cognitive}
\begin{figure}[!tbp]
\begin{center}\setlength{\unitlength}{0.00045833in}
\begingroup\makeatletter\ifx\SetFigFont\undefined%
\gdef\SetFigFont#1#2#3#4#5{%
  \reset@font\fontsize{#1}{#2pt}%
  \fontfamily{#3}\fontseries{#4}\fontshape{#5}%
  \selectfont}%
\fi\endgroup%
{\renewcommand{\dashlinestretch}{30}
\begin{picture}(11863,3829)(0,-10)
\put(8025,2833){\ellipse{212}{212}}
\put(8100,1033){\ellipse{212}{212}}
\put(10500,2833){\ellipse{212}{212}}
\put(10500,1064){\ellipse{212}{212}}
\put(1425,2833){\ellipse{212}{212}}
\put(1500,1033){\ellipse{212}{212}}
\put(3900,2833){\ellipse{212}{212}}
\put(3900,1064){\ellipse{212}{212}}
\path(8100,2833)(10425,2833)
\path(10305.000,2803.000)(10425.000,2833.000)(10305.000,2863.000)
\path(8100,2758)(10425,1108)
\path(10309.777,1152.984)(10425.000,1108.000)(10344.502,1201.915)
\path(8175,1108)(10425,2758)
\path(10345.972,2662.844)(10425.000,2758.000)(10310.490,2711.229)
\path(8250,1033)(10350,1033)
\path(10230.000,1003.000)(10350.000,1033.000)(10230.000,1063.000)
\path(7050,2833)(7500,2833)
\path(7380.000,2803.000)(7500.000,2833.000)(7380.000,2863.000)
\path(7125,1033)(7575,1033)
\path(7455.000,1003.000)(7575.000,1033.000)(7455.000,1063.000)
\path(11025,1033)(11475,1033)
\path(11355.000,1003.000)(11475.000,1033.000)(11355.000,1063.000)
\path(11025,2833)(11475,2833)
\path(11355.000,2803.000)(11475.000,2833.000)(11355.000,2863.000)
\path(1500,2833)(3825,2833)
\path(3705.000,2803.000)(3825.000,2833.000)(3705.000,2863.000)
\path(1500,2758)(3825,1108)
\path(3709.777,1152.984)(3825.000,1108.000)(3744.502,1201.915)
\path(1575,1108)(3825,2758)
\path(3745.972,2662.844)(3825.000,2758.000)(3710.490,2711.229)
\path(1650,1033)(3750,1033)
\path(3630.000,1003.000)(3750.000,1033.000)(3630.000,1063.000)
\path(450,2833)(900,2833)
\path(780.000,2803.000)(900.000,2833.000)(780.000,2863.000)
\dashline{60.000}(675,3583)(675,2833)
\path(645.000,2953.000)(675.000,2833.000)(705.000,2953.000)
\path(525,1033)(975,1033)
\path(855.000,1003.000)(975.000,1033.000)(855.000,1063.000)
\path(780.000,913.000)(750.000,1033.000)(720.000,913.000)
\dashline{60.000}(750,1033)(750,283)
\path(4425,1033)(4875,1033)
\path(4755.000,1003.000)(4875.000,1033.000)(4755.000,1063.000)
\path(4425,2833)(4875,2833)
\path(4755.000,2803.000)(4875.000,2833.000)(4755.000,2863.000)
\dashline{60.000}(11250,3583)(11250,2833)
\path(11220.000,2953.000)(11250.000,2833.000)(11280.000,2953.000)
\path(7380.000,913.000)(7350.000,1033.000)(7320.000,913.000)
\dashline{60.000}(7350,1033)(7350,283)
\put(10650,958){\makebox(0,0)[lb]{{\SetFigFont{7}{8.4}{\rmdefault}{\mddefault}{\updefault}$Y_2$}}}
\put(6600,2758){\makebox(0,0)[lb]{{\SetFigFont{7}{8.4}{\rmdefault}{\mddefault}{\updefault}$W_{1}$}}}
\put(6675,958){\makebox(0,0)[lb]{{\SetFigFont{7}{8.4}{\rmdefault}{\mddefault}{\updefault}$W_{2}$}}}
\put(7575,2758){\makebox(0,0)[lb]{{\SetFigFont{7}{8.4}{\rmdefault}{\mddefault}{\updefault}$X_1$}}}
\put(7650,958){\makebox(0,0)[lb]{{\SetFigFont{7}{8.4}{\rmdefault}{\mddefault}{\updefault}$X_2$}}}
\put(11550,958){\makebox(0,0)[lb]{{\SetFigFont{7}{8.4}{\rmdefault}{\mddefault}{\updefault}$\hat{W}_{2}$}}}
\put(10650,2758){\makebox(0,0)[lb]{{\SetFigFont{7}{8.4}{\rmdefault}{\mddefault}{\updefault}$Y_1$}}}
\put(11550,2758){\makebox(0,0)[lb]{{\SetFigFont{7}{8.4}{\rmdefault}{\mddefault}{\updefault}$\hat{W}_{1}$}}}
\put(4050,958){\makebox(0,0)[lb]{{\SetFigFont{7}{8.4}{\rmdefault}{\mddefault}{\updefault}$Y_2$}}}
\put(450,3658){\makebox(0,0)[lb]{{\SetFigFont{7}{8.4}{\rmdefault}{\mddefault}{\updefault}$W_{2}$}}}
\put(750,3208){\makebox(0,0)[lb]{{\SetFigFont{8}{9.6}{\rmdefault}{\mddefault}{\updefault}Genie}}}
\put(0,2758){\makebox(0,0)[lb]{{\SetFigFont{7}{8.4}{\rmdefault}{\mddefault}{\updefault}$W_{1}$}}}
\put(75,958){\makebox(0,0)[lb]{{\SetFigFont{7}{8.4}{\rmdefault}{\mddefault}{\updefault}$W_{2}$}}}
\put(825,508){\makebox(0,0)[lb]{{\SetFigFont{8}{9.6}{\rmdefault}{\mddefault}{\updefault}Genie}}}
\put(675,58){\makebox(0,0)[lb]{{\SetFigFont{7}{8.4}{\rmdefault}{\mddefault}{\updefault}$W_{1}$}}}
\put(975,2758){\makebox(0,0)[lb]{{\SetFigFont{7}{8.4}{\rmdefault}{\mddefault}{\updefault}$X_1$}}}
\put(1050,958){\makebox(0,0)[lb]{{\SetFigFont{7}{8.4}{\rmdefault}{\mddefault}{\updefault}$X_2$}}}
\put(4950,958){\makebox(0,0)[lb]{{\SetFigFont{7}{8.4}{\rmdefault}{\mddefault}{\updefault}$\hat{W}_{2}$}}}
\put(4050,2758){\makebox(0,0)[lb]{{\SetFigFont{7}{8.4}{\rmdefault}{\mddefault}{\updefault}$Y_1$}}}
\put(4950,2758){\makebox(0,0)[lb]{{\SetFigFont{7}{8.4}{\rmdefault}{\mddefault}{\updefault}$\hat{W}_{1}$}}}
\put(11025,3658){\makebox(0,0)[lb]{{\SetFigFont{7}{8.4}{\rmdefault}{\mddefault}{\updefault}$W_{2}$}}}
\put(11325,3208){\makebox(0,0)[lb]{{\SetFigFont{8}{9.6}{\rmdefault}{\mddefault}{\updefault}Genie}}}
\put(7425,508){\makebox(0,0)[lb]{{\SetFigFont{8}{9.6}{\rmdefault}{\mddefault}{\updefault}Genie}}}
\put(7275,58){\makebox(0,0)[lb]{{\SetFigFont{7}{8.4}{\rmdefault}{\mddefault}{\updefault}$W_{1}$}}}
\end{picture}
}\end{center}
\caption{Two user interference channel with cognitive message sharing, each with $2$ degrees of freedom}
\label{fig:cogint}
\end{figure}
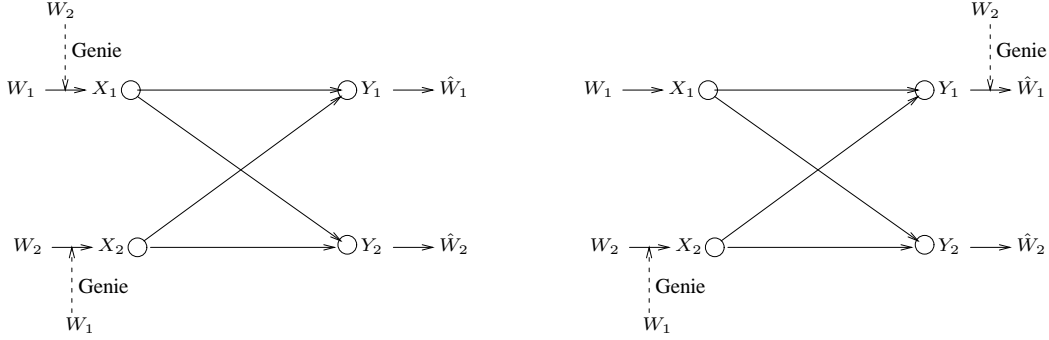
The benefits of cognitive cooperation in communication networks is currently an active area of research \cite{devroye_tarokh:cogradio},\cite{jovicic_pviswanath:cogradio},\cite{lapidoth_shamai:dofcognitiveint}, \cite{wu_svishwanath_etal}. A commonly studied model of cognitive interference networks assumes that certain nodes acquire apriori knowledge of certain messages with the aid of a genie. From the degrees of freedom perspective, cognitive message sharing can indeed increase the number of degrees of freedom of a network. For example, it is easily seen that in the $2$ user interference network with 2 messages (see Figure \ref{fig:cogint}), sharing two messages achieves the full cooperation multiplexing gain of $2$  \cite{jafar_shamai:dofx}. Other examples where cognitive message sharing increases the number of degrees of freedom of interference and $X$ networks can be found in \cite{lapidoth_shamai:dofcognitiveint}, \cite{jafar_shamai:dofx}, \cite{devroye_tarokh:cogx}. 

It is important to note that the models of \cite{jafar_shamai:dofx,lapidoth_shamai:dofcognitiveint,devroye_tarokh:cogx}, do not account for the cost of cooperation. For example, in Figure \ref{fig:cogint} (a), the cost of transmitter $2$ acquiring message $W_1$ and transmitter $1$ acquiring $W_2$ are not factored into the problem formulation. References \cite{host-madsen:dofint, host-madsen_nosratinia:dofint} consider the $2$ user interference channel where the cost-free genie is replaced by physical channels through which the transmitters and/or receivers can share information and find that the network has only $1$ degree of freedom. Thus, the degrees of freedom benefits of genie-aided cooperation disappear when the cost of cooperation is taken into account. Theorem \ref{thm:main} in this paper extends this insight to the general class of $X$ networks.  
\subsection{Full Duplex Operation}
\label{subsec:fd}
The fact that full duplex operation improves the number of degrees of freedom can be easily observed in the $2$ way $2$ user point-to-point channel which has $2$ degrees of freedom (Figure \ref{fig:full_duplex}). In a more general network with $K$ nodes and a message from every node to every other node (Figure \ref{fig:full_duplex}), full duplex operation can be shown to increase the number of degrees of freedom. For example, consider a network with $K$ nodes. 
In Theorem \ref{thm:Kuserfd}, we show that the capacity of this $K$ node full duplex network satisfies
$$ \frac{K(K-1)}{2K-2} \log(\mbox{SNR}) + o(\log(\mbox{SNR})) \leq C_{full-duplex}( \mbox{SNR}) \leq \frac{K(K-1)}{2K-3} \log(\mbox{SNR}) + o(\log(\mbox{SNR}))$$

If we assume half duplex operation, then the optimal arrangement is with $K/2$ nodes behaving as transmitters and $K/2$ nodes behaving as receivers to form a $K/2$ user $X$ network. This network has 
$\frac{K^2}{4K-4}$ degrees of freedom so that we can write 
$$ C_{half-duplex}(\mbox{SNR}) = \frac{K^2}{4K-4} \log(\mbox{SNR}) + o(\log(\mbox{SNR}))$$

Clearly full-duplex operation increases the degrees of freedom of a network with $K$ nodes.

Note that the $S \times R \times D$ network of Theorem \ref{thm:main} includes full-duplex operation in its model. The difference from the $K$ user full duplex network of section \ref{sec:fd} and the $S \times R \times D$ network of Theorem \ref{thm:main} is that in the latter, nodes that are message sources are not destination nodes.

\subsection{Feedback}
\label{subsec:fb}

\begin{figure}[!tbp]
\begin{center}\setlength{\unitlength}{0.00054167in}
\begingroup\makeatletter\ifx\SetFigFont\undefined%
\gdef\SetFigFont#1#2#3#4#5{%
  \reset@font\fontsize{#1}{#2pt}%
  \fontfamily{#3}\fontseries{#4}\fontshape{#5}%
  \selectfont}%
\fi\endgroup%
{\renewcommand{\dashlinestretch}{30}
\begin{picture}(7101,4533)(0,-10)
\put(4875,987){\ellipse{212}{212}}
\put(1469,987){\ellipse{212}{212}}
\put(3119,3237){\ellipse{212}{212}}
\blacken\path(3288.640,3868.066)(3225.000,3762.000)(3331.066,3825.640)(3288.640,3868.066)
\path(3225,3762)(3750,4287)
\path(5550,912)(6075,987)
\path(5960.449,940.331)(6075.000,987.000)(5951.963,999.728)
\blacken\path(5672.115,742.304)(5550.000,762.000)(5648.480,687.155)(5672.115,742.304)
\path(5550,762)(6075,537)
\path(900,837)(600,987)
\path(720.748,960.167)(600.000,987.000)(693.915,906.502)
\blacken\path(887.536,674.536)(975.000,762.000)(856.666,725.985)(887.536,674.536)
\path(975,762)(600,537)
\path(3322.728,3086.176)(3225.000,3162.000)(3275.549,3049.107)
\path(3225,3162)(4875,1062)
\path(4777.272,1137.824)(4875.000,1062.000)(4824.451,1174.893)
\path(3150,3762)(2700,4212)
\path(2806.066,4148.360)(2700.000,4212.000)(2763.640,4105.934)
\dashline{60.000}(4875,837)(4875,12)(1425,12)(1425,837)
\blacken\path(1455.000,717.000)(1425.000,837.000)(1395.000,717.000)(1455.000,717.000)
\dottedline{300}(2730,12)(3660,12)
\path(1695.000,1017.000)(1575.000,987.000)(1695.000,957.000)
\path(1575,987)(4800,987)
\path(4680.000,957.000)(4800.000,987.000)(4680.000,1017.000)
\dashline{60.000}(1329,1003)(729,1453)(2454,3703)(2979,3253)
\blacken\path(2868.365,3308.317)(2979.000,3253.000)(2907.413,3353.873)(2868.365,3308.317)
\dashline{60.000}(3184,3345)(3814,3810)(5614,1395)(4984,960)
\blacken\path(5065.702,1052.870)(4984.000,960.000)(5099.793,1003.496)(5065.702,1052.870)
\path(3027.000,3048.000)(3075.000,3162.000)(2979.000,3084.000)
\path(3075,3162)(1500,1062)
\path(1548.000,1176.000)(1500.000,1062.000)(1596.000,1140.000)
\put(1050,762){\makebox(0,0)[lb]{{\SetFigFont{8}{9.6}{\rmdefault}{\mddefault}{\updefault}1}}}
\put(5325,837){\makebox(0,0)[lb]{{\SetFigFont{8}{9.6}{\rmdefault}{\mddefault}{\updefault}2}}}
\put(3150,3537){\makebox(0,0)[lb]{{\SetFigFont{8}{9.6}{\rmdefault}{\mddefault}{\updefault}3}}}
\put(6150,912){\makebox(0,0)[lb]{{\SetFigFont{8}{9.6}{\rmdefault}{\mddefault}{\updefault}$\hat{W}_{2,1}$}}}
\put(6225,387){\makebox(0,0)[lb]{{\SetFigFont{8}{9.6}{\rmdefault}{\mddefault}{\updefault}$W_{3,2}$}}}
\put(0,1062){\makebox(0,0)[lb]{{\SetFigFont{8}{9.6}{\rmdefault}{\mddefault}{\updefault}$\hat{W}_{1,3}$}}}
\put(3750,4362){\makebox(0,0)[lb]{{\SetFigFont{8}{9.6}{\rmdefault}{\mddefault}{\updefault}$W_{1,3}$}}}
\put(225,312){\makebox(0,0)[lb]{{\SetFigFont{8}{9.6}{\rmdefault}{\mddefault}{\updefault}$W_{2,1}$}}}
\put(2325,4287){\makebox(0,0)[lb]{{\SetFigFont{8}{9.6}{\rmdefault}{\mddefault}{\updefault}$\hat{W}_{3,2}$}}}
\end{picture}
}\end{center}
\caption{A network where feedback increases degrees of freedom}
\label{fig:fb_helps}
\end{figure}
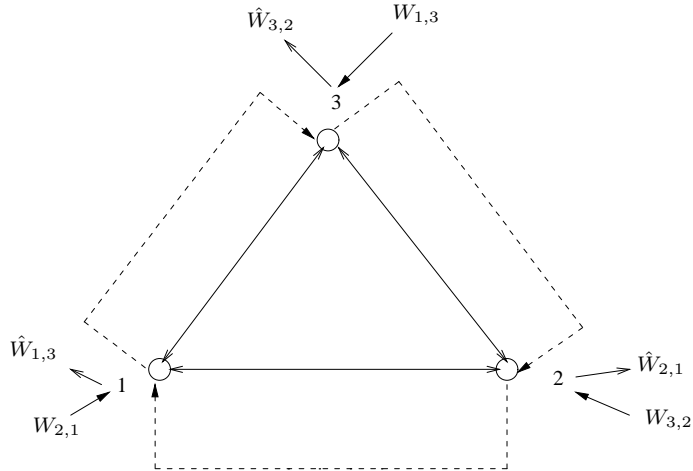

While we know that perfect feedback does not increase the capacity of a memory-less point-to-point channel, feedback has been shown to increase the capacity of the multiple access \cite{ozarow:fbmac},\cite{kramer:fbint} and broadcast \cite{ozarow:fbbc} channels. However, the issue addressed here is whether feedback can increase the \emph{degrees of freedom} of a wireless network. We show in Theorem \ref{thm:main} that even perfect feedback to the source-nodes and relays does not increase the degrees of freedom of a fully connected network with $S$ source nodes, $D$ destination nodes and $R$ relays. We can however construct a scenario where feedback improves the degrees of freedom of a network. To see this consider a $3$ node full duplex network with $3$ messages as in Figure \ref{fig:fb_helps}. If there is no feedback in this $3$ node network, the number of degrees of freedom of this network is upperbounded by $2$ - the degrees of freedom of a $3$ node full-duplex network discussed in the previous subsection. Now, if we include feedback in this network as shown in Figure \ref{fig:fb_helps}, we can show that $3$ degrees of freedom are achievable, with each message achieving $1$ degree of freedom. To see this, consider node $1$. Let the channel gain for the signal from node $i$ to node $j$ be equal to $H_{j,i}$ and $Z_j$ be the AWGN term at receiver $j$. This received symbol at this node is equal to $H_{1,2} X_2 + H_{1,3} X_3 + Z_1$. Also, the node obtains information of $H_{2,3} X_3 + Z_2$ through feedback from node $2$. Using the feedback, node $1$ can zero-force the interference from node $2$ and thus obtain $1$ degree of freedom for the message $W_{1,3}$. Similarly the remaining two messages can each achieve $1$ degree of freedom so that $3$ degrees of freedom are achievable over the network. Thus, feedback increases the degrees of freedom of this network. 

The reason feedback helped the network of Figure \ref{fig:fb_helps} is that it provided a decoder with an extra antenna which can be used to cancel interference. Therefore, feedback can increase the degrees of freedom if it is provided to a decoding node.

\section{System Model for an $S\times R \times D$ node $X$ Network}
\label{sec:SRDnet}
Consider an $S \times R \times D$ node network, i.e., a network with $S+R+D$ nodes where nodes $1,2,\cdots,S$ are sources, nodes $S+1,S+2,\cdots,S+R$ are relays, and nodes $S+R+1,S+R+2,\cdots,S+R+D$ are destination nodes (see Figure \ref{fig:X}). Following the definition of an $X$ network \cite{cadambe_jafar:dofx}, for all $j\in\{1,2,\cdots, S\}$ and for all $i\in\{S+R+1,S+R+2,\cdots,S+R+D\} $, there is an independent message $W_{i,j}$ to be communicated from source node $j$ to destination node $i$.

Full duplex operation is assumed so that all nodes are capable of transmitting and receiving simultaneously. The input and output signals of the $S \times R \times D$ node network are related as:
\begin{eqnarray}\label{eq:defX}
Y_i(n)&=&\sum_{j=1}^{S+R+D}H_{i,j}(n)X_{j}(n)+Z_i(n), ~~~i\in\{1,2,\cdots,S+R+D\}, n\in\mathbb{N}
\end{eqnarray}
where, at the $n^{th}$ discrete time slot, $X_{j}(n)$ is the symbol transmitted by node $j$, $Y_i(n)$ is the symbol received by node $i$, $H_{i,j}(n)$ is the channel from node $j$ to node $i$ and $Z_i(n)$ is the zero mean unit variance additive white Gaussian noise (AWGN) at node $i$. We use the following notation,
\begin{eqnarray}
X_i^n &\triangleq& \{X_i(1),X_i(2),\cdots, X_i(n)\}
\end{eqnarray}
Similar notation is used for output signals and the additive noise terms as well.

The channel coefficients $H_{i,j}(n)$, $\forall i,j\in\{1,2,\cdots,S+D+R\}$ are known \textit{apriori}\footnote{Thus, we also show that non-causal channel knowledge does not increase the degrees of freedom.} to all nodes. We assume the channel coefficients are time-varying/frequency-selective with values drawn randomly from a continuous distribution with support bounded below by a non-zero constant. Thus, all channel coefficients take non-zero values and the network is fully-connected. The AWGN terms $Z_i(n)$ have unit variance and are independent identically distributed (i.i.d.) in time and across nodes.
\begin{figure}[!tbp]
\begin{center}\setlength{\unitlength}{0.00050833in}
\begingroup\makeatletter\ifx\SetFigFont\undefined%
\gdef\SetFigFont#1#2#3#4#5{%
  \reset@font\fontsize{#1}{#2pt}%
  \fontfamily{#3}\fontseries{#4}\fontshape{#5}%
  \selectfont}%
\fi\endgroup%
{\renewcommand{\dashlinestretch}{30}
\begin{picture}(8846,6405)(0,-10)
\path(6825,5250)(6225,4650)
\path(6288.640,4756.066)(6225.000,4650.000)(6331.066,4713.640)
\path(6900,1725)(6255,1972)
\path(6377.793,1957.102)(6255.000,1972.000)(6356.335,1901.070)
\path(975,300)(1350,1725)
\path(1348.473,1601.316)(1350.000,1725.000)(1290.449,1616.586)
\path(3680.527,1618.026)(3626.000,1507.000)(3726.360,1579.304)
\path(3626,1507)(6120,4459)
\path(6065.473,4347.974)(6120.000,4459.000)(6019.640,4386.696)
\path(5993.075,1871.342)(6105.000,1924.000)(5981.463,1930.208)
\path(6105,1924)(3611,1432)
\path(3722.925,1484.658)(3611.000,1432.000)(3734.537,1425.792)
\path(3723.693,5280.351)(3600.000,5280.000)(3709.306,5222.101)
\path(3600,5280)(6090,4665)
\path(5966.307,4664.649)(6090.000,4665.000)(5980.694,4722.899)
\path(1413.803,4420.951)(1384.000,4541.000)(1353.803,4421.049)
\path(1384,4541)(1380,2107)
\path(1350.197,2227.049)(1380.000,2107.000)(1410.197,2226.951)
\path(1546.538,4435.809)(1455.000,4519.000)(1496.621,4402.519)
\path(1455,4519)(3431,1556)
\path(3339.462,1639.191)(3431.000,1556.000)(3389.379,1672.481)
\path(1642.426,4504.791)(1523.000,4537.000)(1613.219,4452.380)
\path(1523,4537)(6045,2017)
\path(5925.574,2049.209)(6045.000,2017.000)(5954.781,2101.620)
\path(1691.076,4577.655)(1568.000,4590.000)(1670.787,4521.190)
\path(1568,4590)(4699,3465)
\path(4575.924,3477.345)(4699.000,3465.000)(4596.213,3533.810)
\path(611,5149)(1361,4699)
\path(1242.666,4735.015)(1361.000,4699.000)(1273.536,4786.464)
\dashline{60.000}(6135,4680)(6113,6094)(1463,6094)(1463,4744)
\path(1433.000,4864.000)(1463.000,4744.000)(1493.000,4864.000)
\path(1695.391,2053.391)(1575.000,2025.000)(1694.588,1993.396)
\path(1575,2025)(6056,1965)
\path(5935.609,1936.609)(6056.000,1965.000)(5936.412,1996.604)
\dashline{60.000}(6150,1920)(6150,825)(1500,825)(1481,1909)
\path(1513.098,1789.544)(1481.000,1909.000)(1453.108,1788.493)
\path(3551.858,1264.218)(3521.000,1384.000)(3491.859,1263.788)
\dashline{60.000}(3521,1384)(3525,825)
\dashline{60.000}(3525,6150)(3525,5325)
\path(3495.000,5445.000)(3525.000,5325.000)(3555.000,5445.000)
\put(375,75){\makebox(0,0)[lb]{\smash{{{\SetFigFont{10}{12.0}{\rmdefault}{\mddefault}{\updefault}$S$ distributed sources}}}}}
\put(4050,150){\makebox(0,0)[lb]{\smash{{{\SetFigFont{10}{12.0}{\rmdefault}{\mddefault}{\updefault}$R$ relays}}}}}
\put(6000,0){\makebox(0,0)[lb]{\smash{{{\SetFigFont{10}{12.0}{\rmdefault}{\mddefault}{\updefault}$D$ distributed destinations}}}}}
\put(4050,6225){\makebox(0,0)[lb]{\smash{{{\SetFigFont{10}{12.0}{\rmdefault}{\mddefault}{\updefault}Perfect Feedback}}}}}
\put(300,5250){\makebox(0,0)[lb]{\smash{{{\SetFigFont{10}{12.0}{\rmdefault}{\mddefault}{\updefault}Node $1$}}}}}
\put(6675,5325){\makebox(0,0)[lb]{\smash{{{\SetFigFont{10}{12.0}{\rmdefault}{\mddefault}{\updefault}Node $S+R+1$}}}}}
\put(6750,1500){\makebox(0,0)[lb]{\smash{{{\SetFigFont{10}{12.0}{\rmdefault}{\mddefault}{\updefault}Node $S+R+D$}}}}}
\put(0,1275){\makebox(0,0)[lb]{\smash{{{\SetFigFont{10}{12.0}{\rmdefault}{\mddefault}{\updefault}Node $S$}}}}}
\path(675,1575)(1350,1950)
\path(1259.670,1865.498)(1350.000,1950.000)(1230.532,1917.948)
\put(6150,4575){\ellipse{212}{212}}
\put(6150,2025){\ellipse{212}{212}}
\put(1455,4633){\ellipse{212}{212}}
\put(4800,3450){\ellipse{212}{212}}
\put(3525,5206){\ellipse{212}{212}}
\put(1519,3417){\ellipse{900}{3600}}
\put(3825,3450){\ellipse{2400}{4500}}
\put(6075,3300){\ellipse{750}{3900}}
\put(1456,2025){\ellipse{220}{220}}
\put(3525,1500){\ellipse{212}{212}}
\path(6204.869,4341.723)(6176.000,4462.000)(6144.871,4342.288)
\path(6176,4462)(6154,2126)
\path(6125.131,2246.277)(6154.000,2126.000)(6185.129,2245.712)
\path(4058.799,1109.018)(3975.000,1200.000)(4006.125,1080.287)
\path(3975,1200)(4425,375)
\path(6366.408,1695.645)(6300.000,1800.000)(6309.487,1676.671)
\path(6300,1800)(6825,225)
\path(1624.744,2180.057)(1534.000,2096.000)(1653.624,2127.465)
\path(1534,2096)(6041,4571)
\path(5950.256,4486.943)(6041.000,4571.000)(5921.376,4539.535)
\path(4958.478,3599.234)(4886.000,3499.000)(4997.120,3553.335)
\path(4886,3499)(6075,4500)
\path(6002.522,4399.766)(6075.000,4500.000)(5963.880,4445.665)
\path(4977.467,3305.714)(4875.000,3375.000)(4932.807,3265.646)
\path(4875,3375)(6053,2062)
\path(5950.533,2131.286)(6053.000,2062.000)(5995.193,2171.354)
\path(4686.840,3257.340)(4725.000,3375.000)(4635.960,3289.140)
\path(4725,3375)(3600,1575)
\path(3638.160,1692.660)(3600.000,1575.000)(3689.040,1660.860)
\path(3522.163,4984.297)(3491.000,5104.000)(3462.166,4983.714)
\path(3491,5104)(3525,1601)
\path(3493.837,1720.703)(3525.000,1601.000)(3553.834,1721.286)
\path(1679.692,1956.503)(1556.000,1957.000)(1664.906,1898.353)
\path(1556,1957)(3428,1481)
\path(3304.308,1481.497)(3428.000,1481.000)(3319.094,1539.647)
\path(1695.198,4679.195)(1575.000,4650.000)(1694.796,4619.196)
\path(1575,4650)(6049,4620)
\path(5928.802,4590.805)(6049.000,4620.000)(5929.204,4650.804)
\path(1608.423,4784.535)(1500.000,4725.000)(1623.684,4726.508)
\path(1500,4725)(3424,5231)
\path(3315.577,5171.465)(3424.000,5231.000)(3300.316,5229.492)
\path(3383.730,5043.046)(3424.000,5160.000)(3333.431,5075.755)
\path(3424,5160)(1451,2126)
\path(1491.270,2242.954)(1451.000,2126.000)(1541.569,2210.245)
\path(3650.641,5021.560)(3555.000,5100.000)(3602.477,4985.781)
\path(3555,5100)(4725,3525)
\path(4629.359,3603.440)(4725.000,3525.000)(4677.523,3639.219)
\path(3722.191,5075.101)(3623.000,5149.000)(3675.746,5037.117)
\path(3623,5149)(6101,2119)
\path(6001.809,2192.899)(6101.000,2119.000)(6048.254,2230.883)
\end{picture}
}\end{center}
\caption{The $S \times R \times D$ network}
\label{fig:X}
\end{figure}
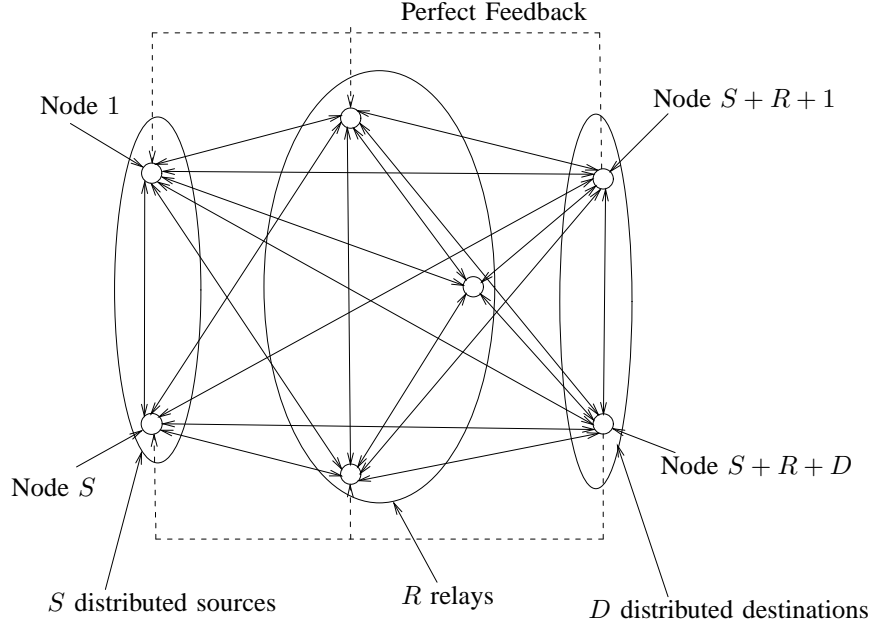

Perfect (noise-free) and causal feedback of all received signals is available to all source and relay nodes, but not to the destination nodes. For codewords spanning $N$ channel uses, the encoding functions are as follows,
\begin{eqnarray*}
X_i(n)=
\left\{
\begin{array}{ll}
f_{i,n}\left(W_{S+R+1,i}, W_{S+R+2,i},\cdots, W_{S+R+D,i},Y_1^{n-1}, Y_2^{n-1},\cdots, Y_{S+R+D}^{n-1}\right), & i\in\{1,2,\cdots, S\}\\
f_{i,n}\left(Y_1^{n-1}, Y_2^{n-1},\cdots, Y_{S+R+D}^{n-1}\right), & i\in\{S+1,S+2,\cdots, S+R\}\\
f_{i,n}\left(Y_i^{n-1}\right), & i\in\{S+R+1,\cdots,S+R+D\}
\end{array}
\right.
\end{eqnarray*}
for $n=1,2,\cdots, N$. 
In other words, the signal transmitted from a source nodes at time $n$ is completely determined by all the messages originating at that source node and the received signals of \emph{all} the nodes upto time $n-1$ (causality condition). The signal transmitted by a destination node at time $n$ can only depend on all the received signals at \emph{that} node upto time $n-1$. This is because the destination nodes do not receive feedback of other nodes' received signals. The signal transmitted from a relay node can only depend on the received signals of \emph{all} the nodes upto time $n-1$.

The decoding functions are as follows,
\begin{eqnarray*}
\hat{W}_{i,j}&=& g_{i,j}\left(Y_i^{N}\right),  i\in\{S+R+1,S+R+2,\cdots, S+R+D\}, j\in\{1,2,\cdots,S\}
\end{eqnarray*}
Thus, a destination node can only use its own received signal to decode all its desired messages. The probability of error is the probability that there is at least one message $W_{i,j}$ that is not decoded correctly, i.e. $\hat{W}_{i,j}\neq W_{i,j}$ for some $(i,j)$.

The total power across all transmitters is assumed to be $\rho$ per channel use. We denote the size of the message set by $|W_{i,j}(\rho)|$. Let $ R_{i,j}(\rho) = \frac{\log| W_{i,j}(\rho)|}{N} $ denote the rate of the codeword encoding the message $W_{i,j}$, where the codewords span $N$ slots. A rate-matrix$[(R_{i,j}(\rho))]$ is said to be \emph{achievable} if messages $W_{i,j}$ can be encoded at rates $R_{i,j}(\rho)$ so that the probability of error can be made arbitrarily small simultaneously for all messages by choosing appropriately long $N$.

Let $C(\rho)$ represent the capacity region of the $S\times R \times D$ node network, i.e., it represents the set of all achievable rate-matrices $[(R_{ji}(\rho))]$. The degrees of freedom region of the $S\times R\times D$ node network is defined as
\begin{eqnarray*}
\mathcal{D} = \Bigg\{ [(d_{i,j})] \in \mathbb{R}_{+}^{SD}: \forall [(\alpha_{i,j})] \in \mathbb{R}_+^{SD} & &\\
		\displaystyle\sum_{i=S+R+1}^{S+R+D} \sum_{j=1}^{S} \alpha_{i,j} d_{i,j} & \leq & \lim\sup_{\rho \to \infty} \left[ \sup_{[(R_{i,j}(\rho))] \in C(\rho)} \displaystyle\sum_{i=S+R+1}^{S+R+D} \sum_{j=1}^{S} \left(\alpha_{i,j} R_{i,j}(\rho) \right)\frac{1}{\log(\rho)} \right] \Bigg\} 
\end{eqnarray*}
\section{Degrees of Freedom of the $S\times R\times D$ node $X$ network}
\label{subsec:main}

\begin{theorem}
\label{thm:main}
Let 
\begin{eqnarray*} \mathcal{D}^{out} \define \bigg\{ [(d_{i,j})] : \forall (u,v) \in \{1,2 \cdots S\} \times \{S+R+1,S+R+2, \cdots, S+R+D\}\\
 \displaystyle\sum_{q=S+R+1}^{S+R+D} d_{q,u} + \displaystyle\sum_{p=1}^{S} d_{v,p} - d_{v,u} & \leq 1 & \bigg\} \end{eqnarray*}
Then $\mathcal{D} \subseteq \mathcal{D}^{out}$ where $\mathcal{D}$ represents the degrees of freedom region of the $S\times R\times D$ node $X$ network.
Furthermore, the total number of degrees of freedom of the $S \times R \times D$ network can be upperbounded as follows
\begin{eqnarray*}
\max_{[(d_{i,j})] \in \mathcal{D}} \sum_{j=1}^{S} \sum_{i=S+R+1}^{S+R+D} d_{i,j} &\leq& \frac{SD}{S+D-1}
\end{eqnarray*}
Equivalently, the sum capacity $C(\rho)$ of this network can be bounded as 
\begin{equation*}
C(\rho) \leq \frac{SD}{S+D-1} \log(\rho) + o(\log(\rho))
\end{equation*}
\end{theorem}
\begin{proof}
\begin{figure}[!tbp]
\begin{center}\setlength{\unitlength}{0.00050000in}
\begingroup\makeatletter\ifx\SetFigFont\undefined%
\gdef\SetFigFont#1#2#3#4#5{%
  \reset@font\fontsize{#1}{#2pt}%
  \fontfamily{#3}\fontseries{#4}\fontshape{#5}%
  \selectfont}%
\fi\endgroup%
{\renewcommand{\dashlinestretch}{30}
\begin{picture}(7800,4962)(0,-10)
\put(6900,987){\ellipse{150}{150}}
\put(2775,3987){\ellipse{150}{150}}
\put(2700,1287){\ellipse{150}{150}}
\put(2700,612){\ellipse{150}{150}}
\put(6900,4287){\ellipse{150}{150}}
\put(6900,3612){\ellipse{150}{150}}
\path(2805.000,3792.000)(2775.000,3912.000)(2745.000,3792.000)
\path(2775,3912)(2775,1437)
\path(2745.000,1557.000)(2775.000,1437.000)(2805.000,1557.000)
\path(6870.000,1257.000)(6900.000,1137.000)(6930.000,1257.000)
\path(6900,1137)(6900,3462)
\path(6930.000,3342.000)(6900.000,3462.000)(6870.000,3342.000)
\dashline{60.000}(2775,4737)(1350,4737)(1350,987)(2550,987)
\blacken\path(2430.000,957.000)(2550.000,987.000)(2430.000,1017.000)(2430.000,957.000)
\dashline{60.000}(6900,4437)(6900,4737)(2775,4737)(2775,4137)
\blacken\path(2745.000,4257.000)(2775.000,4137.000)(2805.000,4257.000)(2745.000,4257.000)
\path(2550,1437)(2850,1437)(2850,462)
	(2550,462)(2550,1437)
\path(6750,4437)(7050,4437)(7050,3462)
	(6750,3462)(6750,4437)
\dottedline{45}(2700,1137)(2700,687)
\dottedline{45}(6900,4137)(6900,3687)
\dashline{60.000}(6900,837)(6900,12)(1875,12)
	(1875,3912)(2625,3912)
\blacken\path(2505.000,3882.000)(2625.000,3912.000)(2505.000,3942.000)(2505.000,3882.000)
\dashline{60.000}(2700,12)(2700,387)
\blacken\path(2730.000,267.000)(2700.000,387.000)(2670.000,267.000)(2730.000,267.000)
\path(7125,3987)(7725,3987)
\path(7605.000,3957.000)(7725.000,3987.000)(7605.000,4017.000)
\path(7050,987)(7650,987)
\path(7530.000,957.000)(7650.000,987.000)(7530.000,1017.000)
\path(975,687)(2475,687)
\blacken\path(2355.000,657.000)(2475.000,687.000)(2355.000,717.000)(2355.000,657.000)
\path(1125,3987)(2625,3987)
\blacken\path(2505.000,3957.000)(2625.000,3987.000)(2505.000,4017.000)(2505.000,3957.000)
\path(2970.000,4017.000)(2850.000,3987.000)(2970.000,3957.000)
\path(2850,3987)(6750,3987)
\path(6630.000,3957.000)(6750.000,3987.000)(6630.000,4017.000)
\path(2965.004,3866.458)(2850.000,3912.000)(2930.043,3817.697)
\path(2850,3912)(6825,1062)
\path(6709.996,1107.542)(6825.000,1062.000)(6744.957,1156.303)
\path(6705.000,957.000)(6825.000,987.000)(6705.000,1017.000)
\path(6825,987)(2850,987)
\path(2970.000,1017.000)(2850.000,987.000)(2970.000,957.000)
\path(2929.186,1157.024)(2850.000,1062.000)(2964.587,1108.580)
\path(2850,1062)(6750,3912)
\path(6670.814,3816.976)(6750.000,3912.000)(6635.413,3865.420)
\put(7275,4287){\makebox(0,0)[lb]{{\SetFigFont{8}{9.6}{\rmdefault}{\mddefault}{\updefault}3}}}
\put(7125,687){\makebox(0,0)[lb]{{\SetFigFont{8}{9.6}{\rmdefault}{\mddefault}{\updefault}4}}}
\put(2550,4137){\makebox(0,0)[lb]{{\SetFigFont{8}{9.6}{\rmdefault}{\mddefault}{\updefault}1}}}
\put(2325,1137){\makebox(0,0)[lb]{{\SetFigFont{8}{9.6}{\rmdefault}{\mddefault}{\updefault}2}}}
\put(3675,4812){\makebox(0,0)[lb]{{\SetFigFont{7}{8.4}{\rmdefault}{\mddefault}{\updefault}Feedback}}}
\put(4500,87){\makebox(0,0)[lb]{{\SetFigFont{7}{8.4}{\rmdefault}{\mddefault}{\updefault}Feedback}}}
\put(7800,3912){\makebox(0,0)[lb]{{\SetFigFont{7}{8.4}{\rmdefault}{\mddefault}{\updefault}$\hat{\overline{W}}_{3,1}$}}}
\put(7725,912){\makebox(0,0)[lb]{{\SetFigFont{7}{8.4}{\rmdefault}{\mddefault}{\updefault}$\hat{\overline{W}}_{4,1}, \hat{\overline{W}}_{4,2}$}}}
\put(0,3912){\makebox(0,0)[lb]{{\SetFigFont{7}{8.4}{\rmdefault}{\mddefault}{\updefault}$\overline{W}_{3,1}, \overline{W}_{4,1}$}}}
\put(375,612){\makebox(0,0)[lb]{{\SetFigFont{7}{8.4}{\rmdefault}{\mddefault}{\updefault}$\overline{W}_{4,2}$}}}
\end{picture}
}\end{center}
\caption{$4$ user $X$ network}
\label{fig:4userX}
\end{figure}
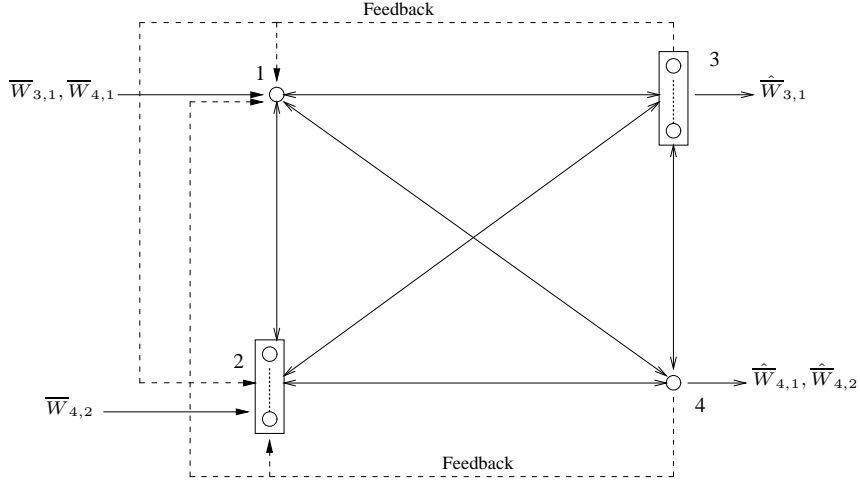
To prove the theorem, all we need to show is that for any $(p,q) \in \{1,2, \ldots, S\} \times \{S+R+1, S+R+2, \ldots, S+R+D\}$
$$\sum_{i=S+R+1}^{S+R+D}d_{i,p}+\sum_{j=1}^{S}d_{q,j}-d_{q,p} \leq 1$$
In other words, for all messages that either originate at node $p$ or are intended for node $q$, the total number of degrees of freedom cannot be more than one.  For convenience, we will show the inequality for $(q,p) = (1,S+R+D)$. By symmetry, the inequality extends to all desired values of $p,q$. We therefore intend to show that 
$$\sum_{i=S+R+1}^{S+R+D}d_{i,1}+\sum_{j=1}^{S}d_{S+R+D,j}-d_{S+R+D,1} \leq 1$$
To show this, we first eliminate all the messages that are not associated with either source node $1$ or destination node $S+R+D$, i.e., we set $W_{i,j} = \phi, \left(i-(S+R+D)\right)(j-1) \neq 0$. Since we are only seeking an outerbound on the rates of a subset of messages, it is important to note that eliminating a message can not hurt the rates of the remaining messages \cite{jafar_shamai:dofx}. Now, we transform the original $S\times R\times D$ node network with single antenna nodes into a $2\times 0 \times 2$ node network, i.e., an $X$ network with $2$ source nodes, zero relay nodes and $2$ destination nodes where one source and one destination have multiple antennas (see Figure \ref{fig:4userX}). This is done by allowing full cooperation between the $S-1$ source nodes $2,\cdots, S$ and the $R$ relay nodes $S+1,S+2,\cdots, S+R$ so that they effectively become one transmitter with $S+R-1$ antennas. Similarly, destination nodes $S+R+1, S+R+2, \cdots, S+R+D-1$ are also allowed to perfectly cooperate so that they form one receiver with $D-1$ antennas. Again, note that allowing the nodes to cooperate cannot reduce the degrees of freedom region and therefore does not contradict our outerbound argument. We represent the resulting $4$ node $X$ network (Figure \ref{fig:4userX}) by the following input-output equations.

\begin{eqnarray}\label{eqn:defXtrans}
\overline{Y}_i(n)&=&\sum_{j=1}^4{\overline{H}}_{i,j}(n)\overline{X}_j(n)+\overline{Z}_i(n), ~~~~~i\in\{1,2,3,4\}
\end{eqnarray}
where 
\begin{eqnarray*}
\overline{Y}_1(n) & = & Y_1(n)\\
\overline{Y}_2(n) & = & \left[Y_2(n)~~Y_3(n)~~\cdots ~~Y_{S+R}(n)\right]^T\\
\overline{Y}_3(n) & = & \left[Y_{S+R+1}(n)~~Y_{S+R+2}(n)~~\cdots ~~Y_{S+R+D-1}(n)\right]\\
\overline{Y}_4(n) & = & Y_{S+R+D}(n)
\end{eqnarray*}
Thus, nodes $2$ and $3$ act as multiple antenna nodes with $S+R-1$ and $D-1$ antennas respectively. $\overline{X}_i(n), \overline{Z}_i(n)$ are also defined in a corresponding manner for $i\in\{1,2,3,4\}$. The definition of the channel coefficients $\overline{H}_{i,j}(n)$ is clear from equations (\ref{eq:defX}) and (\ref{eqn:defXtrans}), and from Figures \ref{fig:X} and \ref{fig:4userX}.
Multiple messages that have the same source and the same destination are combined in the $4$ node $X$ network as follows:
\begin{eqnarray}
\label{eqn:Xmessages1}
\overline{W}_{3,1}&=& \left[ W_{S+R+1,1} ~~W_{S+R+2,1} ~~\cdots ~~ W_{S+R+D-1,1}\right]\\
\overline{W}_{3,2}&=& \phi\\
\overline{W}_{4,1} & = &  W_{S+R+D,1}\\
\overline{W}_{4,2} & = & \left[W_{S+R+D,2} ~~ W_{S+R+D,3} ~~ \cdots ~~ W_{S+R+D,S}\right]
\label{eqn:Xmessages2}
\end{eqnarray}

Over this $X$ network, the encoding functions are as follows
\begin{eqnarray}
\label{eqn:Xencode1}
\overline{X}_1(n)&=&\overline{f}_{1,n}\left(\overline{W}_{3,1},\overline{W}_{4,1},\overline{Y}_1^{n-1},\overline{Y}_2^{n-1},\overline{Y}_3^{n-1},\overline{Y}_4^{n-1}\right)\\
\overline{X}_2(n)&=&\overline{f}_{2,n}\left(\overline{W}_{4,2},\overline{Y}_1^{n-1},\overline{Y}_2^{n-1},\overline{Y}_3^{n-1},\overline{Y}_4^{n-1}\right)\\
\overline{X}_3(n)&=&\overline{f}_{3,n}\left(\overline{Y}_3^{n-1}\right)\\
\overline{X}_4(n)&=&\overline{f}_{4,n}\left(\overline{Y}_4^{n-1}\right) \label{eqn:Xencode2}
\end{eqnarray}
and the decoding functions are the following:
\begin{eqnarray}
\label{eqn:Xdecode1}
\hat{\overline{W}}_{3,1}&=&\overline{g}_{3,1}\left(\overline{Y}_3^{N}\right)\\
\hat{\overline{W}}_{4,1}&=&\overline{g}_{4,1}\left(\overline{Y}_4^{N}\right)\\
\hat{\overline{W}}_{4,2}&=&\overline{g}_{4,2}\left(\overline{Y}_4^{N}\right) \label{eqn:Xdecode2}
\end{eqnarray}
The rates and the degrees of freedom region of this network are defined in a manner similar to the $S \times R \times D$ network.

For this $4$ node $X$ network, it follows from Lemma \ref{lemma:pq} (stated below) that:
$$ \overline{d}_{3,1} + \overline{d}_{4,1} + \overline{d}_{4,2} \leq 1$$
where $\overline{d}_{i,j}$ represents the number of degrees of freedom corresponding to message $\overline{W}_{i,j}$. Using equation (\ref{eqn:Xmessages1})-(\ref{eqn:Xmessages2}), we can re-write the above outerbound in terms of degrees of freedom of the $S \times R \times D$ network as follows.
$$\sum_{i=S+1+R}^{S+R+D}d_{i,1}+\sum_{j=1}^{S}d_{S+R+D,j}-d_{S+R+D,1} \leq 1$$
By symmetry, the above inequality implies that
$$\sum_{i=S+1+R}^{S+R+D}d_{i,p}+\sum_{j=1}^{S}d_{q,j}-d_{q,p} \leq 1$$
The outerbound on the degrees of freedom \emph{region} $\mathcal{D}$ is therefore shown. Summing all inequalities of the above form over all $(p,q) \in \{1,2, \ldots S\} \times \{S+R+1 \ldots S+R+D\}$, the bound on the \emph{total} number of degrees of freedom can be obtained
\end{proof}

In the above proof we used the following lemma.
\begin{lemma} \label{lemma:pq}. 
In the $4$ node $X$ network described by equations (\ref{eqn:defXtrans}),(\ref{eqn:Xencode1})-(\ref{eqn:Xencode2}), (\ref{eqn:Xdecode1})-(\ref{eqn:Xdecode2}) and Figure \ref{fig:4userX}, the total number of degrees of freedom can be upper-bounded by
\begin{eqnarray*}
\max_{\mathcal{D}^{X} } \overline{d}_{3,1} + \overline{d}_{4,2} + \overline{d}_{4,3} \leq 1
\end{eqnarray*}
where $\mathcal{D}^{X}$ represents the degrees of freedom region of this $4$ node $X$ channel
\end{lemma}
Define 
\allowdisplaybreaks{
\begin{eqnarray*}
\overline{U}_i(n) &=& \overline{H}_{i,1}(n)\overline{X}_1(n)+\overline{Z}_i(n),~~i=1,2,3,4
\end{eqnarray*}
}
In order to prove Lemma \ref{lemma:pq}, we need the following result
\begin{lemma}\label{lemma:U}
In the $4$ node $X$ network of Figure \ref{fig:4userX} and equations (\ref{eqn:defXtrans}),(\ref{eqn:Xencode1})-(\ref{eqn:Xencode2}), (\ref{eqn:Xdecode1})-(\ref{eqn:Xdecode2}), the following three statements are true.
\begin{eqnarray*}
S_1(n):& ~~\overline{X}_1^n& \leftarrow \overline{W}_{3,1},\overline{W}_{4,1},\overline{W}_{4,2}, \overline{U}_1^{n-1},\overline{U}_2^{n-1}, \overline{U}_3^{n-1}, \overline{U}_4^{n-1}\\
S_2(n):& ~~\overline{X}_2^n,\overline{X}_3^n,\overline{X}_4^n& \leftarrow \overline{W}_{4,2}, \overline{U}_1^{n-1},\overline{U}_2^{n-1}, \overline{U}_3^{n-1}, \overline{U}_4^{n-1}\\
S_3(n):& ~~\overline{Y}_1^n,\overline{Y}_2^n,\overline{Y}_3^n,\overline{Y}_4^n&\leftarrow \overline{W}_{4,2}, \overline{U}_1^{n},,\overline{U}_2^{n}, \overline{U}_3^{n}, \overline{U}_4^{n}
\end{eqnarray*}
where $A\leftarrow B$ means that the value of $A$ is completely determined by the knowledge of the value taken by $B$.
\end{lemma}

A proof of the above lemma is placed in Appendix \ref{proof:lemmaU}. We now proceed to a proof of Lemma \ref{lemma:pq}

\subsection{Proof of Lemma \ref{lemma:pq}}
\begin{figure}[!tbp]
\begin{center}\setlength{\unitlength}{0.00050000in}
\begingroup\makeatletter\ifx\SetFigFont\undefined%
\gdef\SetFigFont#1#2#3#4#5{%
  \reset@font\fontsize{#1}{#2pt}%
  \fontfamily{#3}\fontseries{#4}\fontshape{#5}%
  \selectfont}%
\fi\endgroup%
{\renewcommand{\dashlinestretch}{30}
\begin{picture}(9611,5058)(0,-10)
\path(8100,4587)(8101,4585)(8102,4581)
	(8105,4574)(8109,4564)(8114,4550)
	(8120,4534)(8126,4514)(8133,4493)
	(8140,4470)(8146,4446)(8153,4420)
	(8158,4392)(8164,4361)(8168,4328)
	(8172,4292)(8174,4253)(8175,4212)
	(8174,4171)(8172,4133)(8168,4101)
	(8164,4074)(8159,4053)(8154,4036)
	(8149,4024)(8143,4014)(8137,4006)
	(8132,3998)(8126,3988)(8121,3977)
	(8116,3963)(8111,3945)(8107,3922)
	(8103,3896)(8101,3867)(8100,3837)
	(8101,3809)(8103,3786)(8107,3766)
	(8111,3749)(8117,3736)(8122,3725)
	(8129,3716)(8135,3709)(8142,3703)
	(8149,3699)(8155,3695)(8161,3692)
	(8166,3690)(8170,3689)(8173,3688)
	(8174,3687)(8175,3687)
\path(8100,2787)(8101,2789)(8102,2793)
	(8105,2800)(8109,2810)(8114,2824)
	(8120,2840)(8126,2860)(8133,2881)
	(8140,2904)(8146,2928)(8153,2954)
	(8158,2982)(8164,3013)(8168,3046)
	(8172,3082)(8174,3121)(8175,3162)
	(8174,3203)(8172,3241)(8168,3273)
	(8164,3300)(8159,3321)(8154,3338)
	(8149,3350)(8143,3360)(8137,3368)
	(8132,3376)(8126,3386)(8121,3397)
	(8116,3411)(8111,3429)(8107,3452)
	(8103,3478)(8101,3507)(8100,3537)
	(8101,3565)(8103,3588)(8107,3608)
	(8111,3625)(8117,3638)(8122,3649)
	(8129,3658)(8135,3665)(8142,3671)
	(8149,3675)(8155,3679)(8161,3682)
	(8166,3684)(8170,3685)(8173,3686)
	(8174,3687)(8175,3687)
\path(9075,4737)(9075,3687)
\blacken\path(9045.000,3807.000)(9075.000,3687.000)(9105.000,3807.000)(9045.000,3807.000)
\path(8175,3687)(9450,3687)
\path(9330.000,3657.000)(9450.000,3687.000)(9330.000,3717.000)
\put(9525,3612){\makebox(0,0)[lb]{{\SetFigFont{7}{8.4}{\rmdefault}{\mddefault}{\updefault}$\hat{\overline{W}}_{3,1}$}}}
\put(8400,4887){\makebox(0,0)[lb]{{\SetFigFont{7}{8.4}{\rmdefault}{\mddefault}{\updefault}$\overline{W}_{4,1}, \overline{W}_{4,2}$}}}
\put(9150,4362){\makebox(0,0)[lb]{{\SetFigFont{7}{8.4}{\rmdefault}{\mddefault}{\updefault}Genie}}}
\put(6900,987){\ellipse{150}{150}}
\put(2775,3987){\ellipse{150}{150}}
\put(2700,1287){\ellipse{150}{150}}
\put(2700,612){\ellipse{150}{150}}
\put(6900,4287){\ellipse{150}{150}}
\put(6900,2937){\ellipse{150}{150}}
\put(6900,4062){\ellipse{150}{150}}
\put(6900,3762){\ellipse{150}{150}}
\put(6900,3462){\ellipse{150}{150}}
\put(6900,3162){\ellipse{150}{150}}
\path(2805.000,3792.000)(2775.000,3912.000)(2745.000,3792.000)
\path(2775,3912)(2775,1437)
\path(2745.000,1557.000)(2775.000,1437.000)(2805.000,1557.000)
\dashline{60.000}(2775,4737)(1350,4737)(1350,987)(2550,987)
\blacken\path(2430.000,957.000)(2550.000,987.000)(2430.000,1017.000)(2430.000,957.000)
\path(2550,1437)(2850,1437)(2850,462)
	(2550,462)(2550,1437)
\dottedline{45}(2700,1137)(2700,687)
\dashline{60.000}(6900,837)(6900,12)(1875,12)
	(1875,3912)(2625,3912)
\blacken\path(2505.000,3882.000)(2625.000,3912.000)(2505.000,3942.000)(2505.000,3882.000)
\dashline{60.000}(2700,12)(2700,387)
\blacken\path(2730.000,267.000)(2700.000,387.000)(2670.000,267.000)(2730.000,267.000)
\path(7050,987)(7650,987)
\path(7530.000,957.000)(7650.000,987.000)(7530.000,1017.000)
\path(1125,3987)(2625,3987)
\blacken\path(2505.000,3957.000)(2625.000,3987.000)(2505.000,4017.000)(2505.000,3957.000)
\dashline{60.000}(6900,4437)(6900,4737)(2775,4737)(2775,4137)
\blacken\path(2745.000,4257.000)(2775.000,4137.000)(2805.000,4257.000)(2745.000,4257.000)
\path(2850,3987)(6825,4287)
\path(6707.598,4248.054)(6825.000,4287.000)(6703.083,4307.884)
\path(2925,3987)(6825,2937)
\path(6701.327,2939.228)(6825.000,2937.000)(6716.925,2997.165)
\dashline{60.000}(6900,3987)(6900,3837)
\path(2850,3987)(6900,3912)
\path(6779.465,3884.227)(6900.000,3912.000)(6780.576,3944.217)
\dashline{60.000}(6900,3387)(6900,3237)
\path(3048.329,3996.487)(2925.000,3987.000)(3038.284,3937.334)
\path(2925,3987)(6900,3312)
\path(6776.671,3302.513)(6900.000,3312.000)(6786.716,3361.666)
\path(6750,4437)(7050,4437)(7050,2787)
	(6750,2787)(6750,4437)
\path(2964.433,3940.041)(2850.000,3987.000)(2928.872,3891.715)
\path(2850,3987)(6825,1062)
\path(6710.567,1108.959)(6825.000,1062.000)(6746.128,1157.285)
\path(2970.000,1017.000)(2850.000,987.000)(2970.000,957.000)
\path(2850,987)(6825,987)
\path(6705.000,957.000)(6825.000,987.000)(6705.000,1017.000)
\path(1050,687)(2550,687)
\blacken\path(2430.000,657.000)(2550.000,687.000)(2430.000,717.000)(2430.000,657.000)
\put(7125,687){\makebox(0,0)[lb]{{\SetFigFont{8}{9.6}{\rmdefault}{\mddefault}{\updefault}4}}}
\put(2550,4137){\makebox(0,0)[lb]{{\SetFigFont{8}{9.6}{\rmdefault}{\mddefault}{\updefault}1}}}
\put(2325,1137){\makebox(0,0)[lb]{{\SetFigFont{8}{9.6}{\rmdefault}{\mddefault}{\updefault}2}}}
\put(3675,4812){\makebox(0,0)[lb]{{\SetFigFont{7}{8.4}{\rmdefault}{\mddefault}{\updefault}Feedback}}}
\put(4500,87){\makebox(0,0)[lb]{{\SetFigFont{7}{8.4}{\rmdefault}{\mddefault}{\updefault}Feedback}}}
\put(7725,912){\makebox(0,0)[lb]{{\SetFigFont{7}{8.4}{\rmdefault}{\mddefault}{\updefault}$\hat{\overline{W}}_{4,1}, \hat{\overline{W}}_{4,2}$}}}
\put(0,3912){\makebox(0,0)[lb]{{\SetFigFont{7}{8.4}{\rmdefault}{\mddefault}{\updefault}$\overline{W}_{3,1}, \overline{W}_{4,1}$}}}
\put(7200,2862){\makebox(0,0)[lb]{{\SetFigFont{7}{8.4}{\rmdefault}{\mddefault}{\updefault}$\overline{U}_4(n)$}}}
\put(7200,3237){\makebox(0,0)[lb]{{\SetFigFont{7}{8.4}{\rmdefault}{\mddefault}{\updefault}$\overline{U}_3(n)$}}}
\put(7200,4287){\makebox(0,0)[lb]{{\SetFigFont{7}{8.4}{\rmdefault}{\mddefault}{\updefault}$\overline{U}_1(n)$}}}
\put(7200,3912){\makebox(0,0)[lb]{{\SetFigFont{7}{8.4}{\rmdefault}{\mddefault}{\updefault}$\overline{U}_2(n)$}}}
\put(450,612){\makebox(0,0)[lb]{{\SetFigFont{7}{8.4}{\rmdefault}{\mddefault}{\updefault}$\overline{W}_{4,2}$}}}
\end{picture}
}\end{center}
\caption{The outerbound argument for the $4$ user $X$ network - Lemma \ref{lemma:pq} }
\label{fig:4userXconverse}
\end{figure}
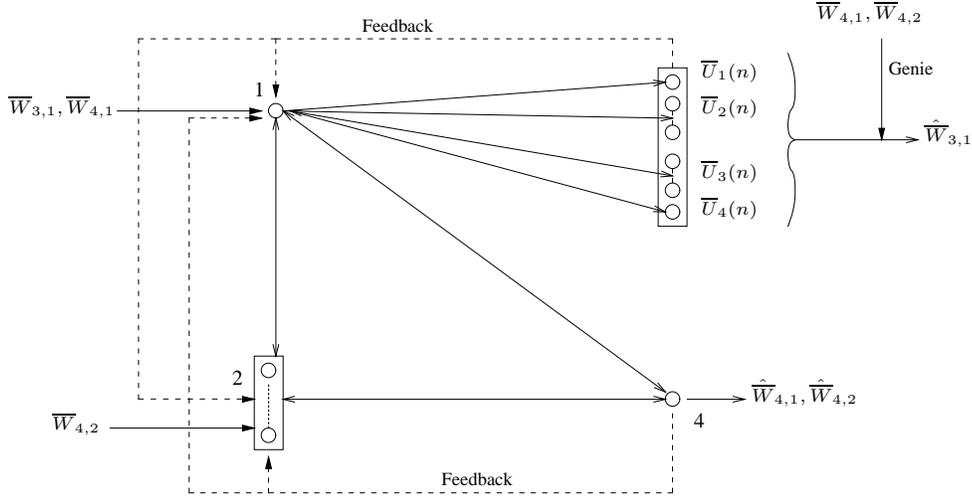

Let a genie provide the messages $\overline{W}_{4,1},\overline{W}_{4,2}$ and $\overline{U}_1(n),\overline{U}_2(n), \overline{U}_3(n), \overline{U}_4^(n)$ to node $3$ at channel use $n \in\mathbb{N}$. Note that Lemma \ref{lemma:U} implies that the node can construct $\overline{Y}_i(n),i=1,2,3,4$ using this side information from the genie. Also note that in the genie supported $4$ node $X$ network in Figure \ref{fig:4userXconverse}, node $3$ is not explicitly shown to receive $\overline{Y}_3(n)$ because it is already contained in the information available to node $3$.

Next we find outerbounds on the rates in the genie supported $4$ node $X$ network. Using Fano's inequality, we can write
\begin{eqnarray}
N\overline{R}_{3,1}(\rho)&\leq& I(\overline{W}_{3,1};\overline{W}_{4,1},\overline{W}_{4,2},\overline{U}_1^{N},\overline{U}_2^{N}, \overline{U}_3^{N}, \overline{U}_4^{N})+N\epsilon_N\nonumber\\
&\leq& I(\overline{W}_{3,1};\overline{U}_1^{N},\overline{U}_2^{N}, \overline{U}_3^{N}, \overline{U}_4^{N}|\overline{W}_{4,1},\overline{W}_{4,2})+N\epsilon_N \nonumber \\
&\leq& \underbrace{H(\overline{U}_1^{N},\overline{U}_2^{N}, \overline{U}_3^{N}, \overline{U}_4^{N}|\overline{W}_{4,1},\overline{W}_{4,2})}_{T_1}-\underbrace{H(\overline{U}_1^{N},\overline{U}_2^{N}, \overline{U}_3^{N}, \overline{U}_4^{N}|\overline{W}_{3,1},\overline{W}_{4,1},\overline{W}_{4,2})}_{T_2}+N\epsilon_N \label{eqn:T1andT2}
\end{eqnarray}

Simplifying the terms,
\allowdisplaybreaks{
\begin{eqnarray}
T_2&=&\sum_{n=1}^N H(\overline{U}_1(n),\overline{U}_2(n), \overline{U}_3(n), \overline{U}_4(n)|\overline{W}_{3,1},\overline{W}_{4,1},\overline{W}_{4,2},\overline{U}_1^{n-1},\overline{U}_2^{n-1}, \overline{U}_3^{n-1}, \overline{U}_4^{n-1}) \nonumber\\
&\stackrel{(a)}{\geq}&\sum_{n=1}^N H(\overline{U}_1(n),\overline{U}_2(n), \overline{U}_3(n), \overline{U}_4(n)|\overline{W}_{3,1},\overline{W}_{4,1},\overline{W}_{4,2},\overline{U}_1^{n-1},\overline{U}_2^{n-1}, \overline{U}_3^{n-1}, \overline{U}_4^{n-1}, \overline{X}_1(n)) \nonumber\\
&\geq&\sum_{n=1}^N H(\overline{Z}_1(n),\overline{Z}_2(n), \overline{Z}_3(n), \overline{Z}_4(n)) \nonumber\\
&\geq& \sum_{n=1}^N\sum_{i=1}^{M_1+M_2+M_3+M_4}H(Z_i(n)) \nonumber\\
&\geq& N(M_1+M_2+M_3+M_4)\log(2\pi e) \label{eqn:T2}
\end{eqnarray}
where $M_i$ represents the number of antennas at node $i$ in the $X$ network. Inequality (a) uses the fact that conditioning on $X_1(n)$ reduces entropy. 

\begin{eqnarray}
T_1 &=&H(\overline{U}_1^{N},\overline{U}_2^{N}, \overline{U}_3^{N}, \overline{U}_4^{N}|\overline{W}_{4,1},\overline{W}_{4,2}) \nonumber\\
&=& \sum_{n=1}^N H(\overline{U}_1(n),\overline{U}_2(n), \overline{U}_3(n), \overline{U}_4(n)|\overline{W}_{4,1},\overline{W}_{4,2},\overline{U}_1^{n-1},\overline{U}_2^{n-1}, \overline{U}_3^{n-1}, \overline{U}_4^{n-1})\nonumber\\
&=& \sum_{n=1}^N H(\overline{U}_4(n)|\overline{W}_{4,1},\overline{W}_{4,2},\overline{U}_1^{n-1},\overline{U}_2^{n-1}, \overline{U}_3^{n-1}, \overline{U}_4^{n-1}) \nonumber\\
&&+ \sum_{n=1}^N H(\overline{U}_1(n), \overline{U}_2(n), \overline{U}_3(n)|\overline{W}_{4,1},\overline{W}_{4,2},\overline{U}_1^{n-1},\overline{U}_2^{n-1}, \overline{U}_3^{n-1}, \overline{U}_4^{n-1},\overline{U}_4(n))\nonumber\\
&\stackrel{(b)}{=}& \sum_{n=1}^N H(\overline{U}_4(n)+ \sum_{j=2}^4\overline{H}_{4,j}\overline{X}_j(n)|\overline{W}_{4,1},\overline{W}_{4,2},\overline{U}_1^{n-1},\overline{U}_2^{n-1}, \overline{U}_3^{n-1}, \overline{U}_4^{n-1}) \nonumber\\
&& + \sum_{n=1}^N H(\overline{U}_1(n), \overline{U}_2(n), \overline{U}_3(n)|\overline{W}_{4,1},\overline{W}_{4,2},\overline{U}_1^{n-1},\overline{U}_2^{n-1}, \overline{U}_3^{n-1}, \overline{U}_4^{n-1},\overline{U}_4(n)) \nonumber\\
&\stackrel{(c)}{=}& \sum_{n=1}^N H(\overline{Y}_4(n)|\overline{W}_{4,1},\overline{W}_{4,2},\overline{U}_1^{n-1},\overline{U}_2^{n-1}, \overline{U}_3^{n-1}, \overline{U}_4^{n-1}, \overline{Y}_4(n-1)) \nonumber\\
&&\sum_{n=1}^N H(\overline{U}_1(n), \overline{U}_2(n), \overline{U}_3(n)|\overline{W}_{4,1},\overline{W}_{4,2},\overline{U}_1^{n-1},\overline{U}_2^{n-1}, \overline{U}_3^{n-1}, \overline{U}_4^{n-1},\overline{U}_4(n)) \nonumber\\
&\stackrel{(d)}{\leq}& \sum_{n=1}^N H(\overline{Y}_4(n)|\overline{W}_{4,1},\overline{W}_{4,2},\overline{Y}_4(n-1)) + \sum_{n=1}^N H(\overline{U}_1(n), \overline{U}_2(n), \overline{U}_3(n)|\overline{U}_4(n)) \nonumber \\
&\leq& H(\overline{Y}_4^N|\overline{W}_{4,1}\overline{W}_{4,2}) + \sum_{n=1}^N \sum_{i=1}^{3} H(\overline{U}_i(n)|\overline{U}_4(n)) \label{eqn:T1}
\end{eqnarray}
}
Equality (b) uses statements $S_1(n), S_2(n)$ of Lemma \ref{lemma:U}. (c) uses $S_3(n)$ of Lemma \ref{lemma:U}. Inequality (d) uses the fact that conditioning reduces entropy to bound both terms of the summand on the right hand side.

Now, we can bound the second term in $T_1$ as follows
\begin{eqnarray}
H(\overline{U}_i(n)|\overline{U}_4(n)) &=&  H(\overline{H}_{i,1}(n)\overline{X}_1(n)+\overline{Z}_{i}(n)|\overline{H}_{,1}(n)X_1(n)+Z_{S+R+D}(n)) \nonumber \\
	& \stackrel{(e)}{\leq} & \displaystyle \sum_{j=1}^{M_i} H(\overline{H}_{i,1}^{[j]} \overline{X}_1(n) + \overline{Z}_{i}^{[j]}(n)| \overline{H}_{4,1}(n)\overline{X}_1(n)+\overline{Z}_{4}(n)) \nonumber \\
	& \stackrel{(f)}{\leq} & \displaystyle \sum_{j=1}^{M_i}\left(\log\left(1+\frac{|H_{i,1}^{[j]}|^2 \rho_{1} }{1+|\overline{H}_{4,1}|^2 \rho_1 } \right) + \log(2\pi e)\right) \nonumber \\ 
	& \leq & \displaystyle \sum_{j=1}^{M_i}\left(\log\left(1+\frac{|H_{i,1}^{[j]}|^2 \rho_{1} }{1+|\overline{H}_{4,1}|^2 \rho_1} \right) \right) + NM_i\log(2\pi e) \label{eqn:T1part} 
\end{eqnarray}
where in (e), $\overline{Z}_i^{[j]}$ and $\overline{H}_{i,1}^{[j]}$ represent the noise and channel gain terms associated with the $j$th antenna at node $i$ and $M_i$ represents the number of antennas at node $i$ (Note that $M_1=M_3=1$ ). Inequality ($f$) above holds because of Lemma $1$ in \cite{host-madsen:dofint}. Note that $\rho_1$ represents the power expended at transmitter $1$. Now, notice that the bound of (\ref{eqn:T1part}) implies that as long as $\overline{H}_{4,1} > 0$, the term $H_{\overline{U}_i(n)|\overline{U}_4(n)}$ is upper-bounded by a constant for all values of input power $\rho$.
Therefore, using equation (\ref{eqn:T1part}) and then combining (\ref{eqn:T1}), (\ref{eqn:T2}) and (\ref{eqn:T1andT2}) we can write 
\begin{eqnarray}
\overline{R}_{3,1}(\rho)\leq \frac{1}{N}H(\overline{Y}_4^N|\overline{W}_{4,1},\overline{W}_{4,2}) + O(1) \label{eqn:R3bound}
\end{eqnarray}

Now, using Fano's inequality to bound rates of messages intended for node $4$, we can write 
\begin{eqnarray}
\overline{R}_{4,1}(\rho)+\overline{R}_{4,2}(\rho)&\leq& \frac{1}{N}I(\overline{W}_{4,1}, \overline{W}_{4,2};\overline{Y}_4^N)+\epsilon_N \nonumber\\
&=& \frac{1}{N}H(\overline{Y}_4^N)-\frac{1}{N}H(\overline{Y}_4^N|\overline{W}_{4,1},\overline{W}_{4,2}) +\epsilon_N \nonumber\\
&=& \frac{1}{N}(\sum_{j=1}^{3} \overline{H}_{4,j} \overline{X}_j)-\frac{1}{N}H(\overline{Y}_4^N|\overline{W}_{4,1},\overline{W}_{4,2}) +\epsilon_N \nonumber\\
&=& \log(\rho) + o(\log(\rho)) - \frac{1}{N}H(\overline{Y}_4^N|\overline{W}_{4,1},\overline{W}_{4,2}) + \epsilon_N \label{eqn:R4bound}
\end{eqnarray}
where the final inequality can be derived from the standard upper-bound on entropy using Gaussian variables.
Adding up the upperbounds (\ref{eqn:R3bound}) and (\ref{eqn:R4bound}), we have
\begin{eqnarray}
\overline{R}_{3,1}(\rho)+\overline{R}_{4,1}(\rho)+\overline{R}_{4,2}(\rho)&\leq& \log(\rho) + o(\log(\rho)) + \epsilon_N
\end{eqnarray}
Thus, the total number of degrees of freedom of the $4$ node $X$ network described is upper-bounded by $1$ so that we can write
\begin{eqnarray*}
\max_{\mathcal{D}^{X} } \overline{d}_{3,1} + \overline{d}_{4,2} + \overline{d}_{4,3} \leq 1
\end{eqnarray*}
\QED

Note that the above lemma holds even if the $X$ network is not fully connected long as $H_{4,1}$ is non-zero. If $H_{4,1}$ is equal to zero,the argument fails because in equation (\ref{eqn:T1part}), $H(U_1(n),U_2(n),U_3(n)|U_4(n))$ cannot be upperbounded by $o(\log(\rho))$. All other inequalities hold for arbitrary channel co-efficients.

\section{$K$ user Full Duplex Network}
\label{sec:fd}
In this section, we derive bounds on the degrees of freedom region of the $K$-user full-duplex network (see Figure \ref{fig:Kuserfd} (a)).
Consider a fully connected network with $K$ full-duplex nodes $1,2, \ldots K$. In this network assume that there exists a message from every node to every other node in the network so that there are a total of $K(K-1)$ messages in the system. The message from node $j$ to node $i \neq j$ is denoted by $W_{j,i}$. Let $H_{i,j}(n)$ represent the channel gain between nodes $i$ and $j$ corresponding to the $n$th symbol. The channel gains satisfy $H_{i,j}(n) = H_{j,i}(n)$ and $H_{i,i} = 0$. As usual, all nodes have apriori knowledge of all channel gains. The input-output relations in this channel are represented by 
\begin{eqnarray}\label{eq:def_fd}
Y_i(n)&=&\sum_{j=1}^K H_{i,j}(n)X_j(n)+Z_i(n), ~~~~~i\in\{1,2\ldots K\}
\end{eqnarray}
where $Y_i(n), X_i(n), Z_i(n)$ represent, respectively, the received symbol, the transmitted symbol and the AWGN term at node $i$. 
For codewords of length $N$, the encoding functions in this network are defined as 
\begin{equation} \label{eqn:Kuser_enc}
X_i(n) =  f_{i,n}(W_{1,i}, W_{2,i}, \ldots ,W_{i-1,i}, W_{i+1,i}, \ldots W_{K,i}, Y_i^{n-1} )
\end{equation}
and the decoding functions are defined as 
\begin{equation} \label{eqn:Kuser_dec}
\hat{W}_{j,i} =  g_{j,i}(Y_j^{N},W_{1,j}, W_{2,j}, \ldots ,W_{j-1,j}, W_{j+1,j}, \ldots W_{K,j}), \forall i \neq j.
\end{equation}
The main result of this section is an approximation of the capacity of the $K$ user full duplex network as follows.
\begin{theorem}
\label{thm:Kuserfd}
The capacity $C_{fd}(\rho)$ of the $K$ user full-duplex network is bounded as follows
\begin{eqnarray*} C_{fd}(\rho) \geq \frac{K(K-1)}{2K-2} \log(\rho) + o(\log(\rho)) \\ 
  C_{fd}(\rho) \leq \frac{K(K-1)}{2K-3}\log(\rho) + o(\log(\rho)) \end{eqnarray*}
\end{theorem}

 The outerbound of the above theorem is proved in Section \ref{subsec:fdouterbound}. A discussion of the innerbound is placed in Section \ref{subsec:fdinnerbound}. A formal proof of the innerbound is placed in Appendix \ref{app:Kuserfd_ach}. The proofs of the both the outerbound and innerbound need the lemma below, which transforms the $K$ user full duplex network to another network whose source nodes are different from destination nodes. 

\begin{lemma}
\label{lemma:Kusereq}
The $K$ user full-duplex network is equivalent to a network with $K$ half-duplex source nodes and $K$ half-duplex destination nodes with the following properties
\begin{enumerate}
\item The input-output relations are described as 
\begin{eqnarray*}
\widetilde{Y}_i(n)&=&\sum_{j=1}^K \widetilde{H}_{i,j}(n)\widetilde{X}_j(n)+\widetilde{Z}_i(n), ~~~~~i\in\{1,2\ldots K\}
\end{eqnarray*}
\item \begin{eqnarray*}\widetilde{H}_{i,j} &=& H_{i,j}, \forall i,j=1,2 \ldots K
\end{eqnarray*}
Note that this means 
\begin{eqnarray*}\widetilde{H}_{i,j} &=& \widetilde{H}_{j,i}\\
\widetilde{H}_{i,i} &=& 0
\end{eqnarray*}
\item There are $K(K-1)$ messages in the system, denoted by $\widetilde{W}_{j,i}, i\neq j$. These messages are denoted by 
\begin{eqnarray*}\widetilde{W}_{j,i} &=& W_{j,i}, \forall i \neq j, i,j \in \{1,2, \ldots K\}\end{eqnarray*}
\item Encoding function of the form 
\begin{equation}
\label{eqn:Kuseraltenc}
\widetilde{X}_i(n) = \widetilde{f}_{i,n}(\widetilde{Y}_i^{n-1},\widetilde{W}_{1,i}, \widetilde{W}_{2,i} \ldots \widetilde{W}_{i-1,i}, \widetilde{W}_{i+1,i} \ldots \widetilde{W}_{K,i})
\end{equation}
\item Decoding function of the form
\begin{equation}
\label{eqn:Kuseraltdec}
\hat{\widetilde{W}}_{j,i} = \widetilde{g}_{j,i}(\widetilde{Y}_j^N,\widetilde{W}_{1,j}, \widetilde{W}_{2,j} \ldots \widetilde{W}_{j-1,j}, \widetilde{W}_{j+1,j} \ldots \widetilde{W}_{K,j}), j \neq i
\end{equation}
\end{enumerate}
\end{lemma}
Note that the encoding and decoding functions imply that 
\begin{itemize}
\item A genie provides receiver $j$ with apriori knowledge of all messages at source $j$ i.e. $\widetilde{W}_{i,j}, \forall i=\{1,2, \ldots K\} - \{j\}$
\item There is perfect feedback from destination $K$ to source $K$. 
\end{itemize}
By comparing encoding equations (\ref{eqn:Kuser_enc}), (\ref{eqn:Kuseraltenc}) and decoding equations (\ref{eqn:Kuser_dec}), (\ref{eqn:Kuseraltdec}), the lemma can easily be proved, i.e., we can show that any encoding scheme that can be implemented on the $K$ user full duplex network, can also be implemented on network described in the above lemma and vice-versa.

\begin{figure}[!tbp]
\begin{center}\input{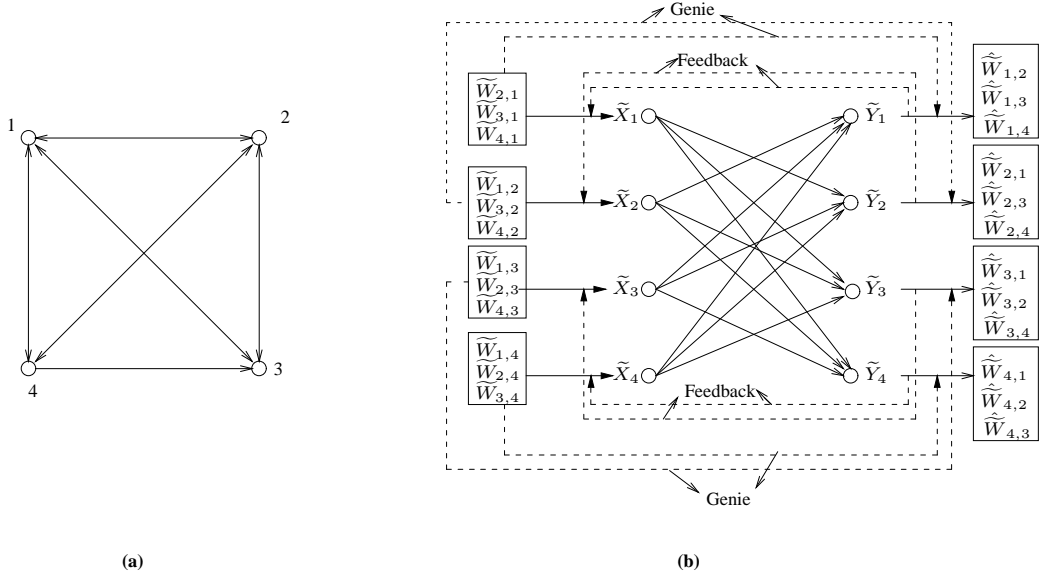}\end{center}
\caption{(a) represents the $4$ user full duplex network. (b) represent an equivalent network}
\label{fig:Kuserfd}
\end{figure}

\begin{figure}[!tbp]
\begin{center}\setlength{\unitlength}{0.00050000in}
\begingroup\makeatletter\ifx\SetFigFont\undefined%
\gdef\SetFigFont#1#2#3#4#5{%
  \reset@font\fontsize{#1}{#2pt}%
  \fontfamily{#3}\fontseries{#4}\fontshape{#5}%
  \selectfont}%
\fi\endgroup%
{\renewcommand{\dashlinestretch}{30}
\begin{picture}(6012,4533)(0,-10)
\put(75,3237){\makebox(0,0)[lb]{{\SetFigFont{7}{8.4}{\rmdefault}{\mddefault}{\updefault}$\overline{W}_{3,1}$}}}
\put(75,3012){\makebox(0,0)[lb]{{\SetFigFont{7}{8.4}{\rmdefault}{\mddefault}{\updefault}$\overline{W}_{4,1}$}}}
\path(75,3537)(600,3537)(600,3012)
	(75,3012)(75,3537)
\put(1425,3162){\makebox(0,0)[lb]{{\SetFigFont{7}{8.4}{\rmdefault}{\mddefault}{\updefault}$\overline{X}_1$}}}
\put(5475,237){\makebox(0,0)[lb]{{\SetFigFont{7}{8.4}{\rmdefault}{\mddefault}{\updefault}$\hat{\overline{W}}_{4,2}$}}}
\put(5475,537){\makebox(0,0)[lb]{{\SetFigFont{7}{8.4}{\rmdefault}{\mddefault}{\updefault}$\hat{\overline{W}}_{4,1}$}}}
\put(1853,2315){\ellipse{150}{150}}
\put(3975,3237){\ellipse{150}{150}}
\put(1897,3215){\ellipse{150}{150}}
\put(1875,1437){\ellipse{150}{150}}
\put(1875,537){\ellipse{150}{150}}
\put(3975,537){\ellipse{150}{150}}
\put(3975,1437){\ellipse{150}{150}}
\put(3975,2337){\ellipse{150}{150}}
\path(1950,3237)(3900,2412)
\path(3777.795,2431.128)(3900.000,2412.000)(3801.173,2486.386)
\path(1950,2337)(3900,3237)
\path(3803.617,3159.474)(3900.000,3237.000)(3778.473,3213.952)
\path(1950,1437)(3900,3162)
\path(3829.998,3060.021)(3900.000,3162.000)(3790.243,3104.961)
\path(1950,612)(3975,3162)
\path(3923.867,3049.370)(3975.000,3162.000)(3876.881,3086.683)
\path(1950,3162)(3900,612)
\path(3803.275,689.099)(3900.000,612.000)(3850.937,725.546)
\path(1950,2262)(3900,537)
\path(3790.243,594.039)(3900.000,537.000)(3829.998,638.979)
\path(1950,1362)(3900,462)
\path(3778.473,485.048)(3900.000,462.000)(3803.617,539.526)
\path(1950,462)(3900,1362)
\path(3803.617,1284.474)(3900.000,1362.000)(3778.473,1338.952)
\path(1950,2262)(3900,1437)
\path(3777.795,1456.128)(3900.000,1437.000)(3801.173,1511.386)
\path(1950,1362)(3900,2337)
\path(3806.085,2256.502)(3900.000,2337.000)(3779.252,2310.167)
\path(1950,537)(3900,2262)
\path(3829.998,2160.021)(3900.000,2262.000)(3790.243,2204.961)
\path(1725,2487)(2025,2487)(2025,312)
	(1725,312)(1725,2487)
\path(1950,3162)(3900,1512)
\path(3789.015,1566.611)(3900.000,1512.000)(3827.772,1612.415)
\path(3825,3462)(4125,3462)(4125,1287)
	(3825,1287)(3825,3462)
\dashline{60.000}(4650,4287)(4650,2112)
\blacken\path(4620.000,2232.000)(4650.000,2112.000)(4680.000,2232.000)(4620.000,2232.000)
\path(525,1437)(1275,1437)
\blacken\path(1155.000,1407.000)(1275.000,1437.000)(1155.000,1467.000)(1155.000,1407.000)
\path(600,3237)(1350,3237)
\blacken\path(1230.000,3207.000)(1350.000,3237.000)(1230.000,3267.000)(1230.000,3207.000)
\path(4500,2112)(5550,2112)
\path(5430.000,2082.000)(5550.000,2112.000)(5430.000,2142.000)
\path(4425,537)(5475,537)
\path(5355.000,507.000)(5475.000,537.000)(5355.000,567.000)
\path(5475,837)(6000,837)(6000,237)
	(5475,237)(5475,837)
\dashline{60.000}(1050,3237)(1050,1437)
\blacken\path(1020.000,1557.000)(1050.000,1437.000)(1080.000,1557.000)(1020.000,1557.000)
\dashline{60.000}(4725,462)(4725,12)(900,12)(900,1437)
\blacken\path(930.000,1317.000)(900.000,1437.000)(870.000,1317.000)(930.000,1317.000)
\dashline{60.000}(5025,2112)(5025,3687)(1050,3687)(1050,3237)
\blacken\path(1020.000,3357.000)(1050.000,3237.000)(1080.000,3357.000)(1020.000,3357.000)
\dashline{60.000}(900,1437)(900,3237)
\blacken\path(930.000,3117.000)(900.000,3237.000)(870.000,3117.000)(930.000,3117.000)
\put(2400,87){\makebox(0,0)[lb]{{\SetFigFont{7}{8.4}{\rmdefault}{\mddefault}{\updefault}Feedback}}}
\put(2550,3762){\makebox(0,0)[lb]{{\SetFigFont{7}{8.4}{\rmdefault}{\mddefault}{\updefault}Feedback}}}
\put(4575,4362){\makebox(0,0)[lb]{{\SetFigFont{7}{8.4}{\rmdefault}{\mddefault}{\updefault}$\overline{W}_{4,2}$}}}
\put(4725,3987){\makebox(0,0)[lb]{{\SetFigFont{7}{8.4}{\rmdefault}{\mddefault}{\updefault}Genie}}}
\put(1350,1362){\makebox(0,0)[lb]{{\SetFigFont{7}{8.4}{\rmdefault}{\mddefault}{\updefault}$\overline{X}_2$}}}
\put(0,1362){\makebox(0,0)[lb]{{\SetFigFont{7}{8.4}{\rmdefault}{\mddefault}{\updefault}$\overline{W}_{4,2}$}}}
\put(4200,2037){\makebox(0,0)[lb]{{\SetFigFont{7}{8.4}{\rmdefault}{\mddefault}{\updefault}$\overline{Y}_3$}}}
\put(4125,462){\makebox(0,0)[lb]{{\SetFigFont{7}{8.4}{\rmdefault}{\mddefault}{\updefault}$\overline{Y}_4$}}}
\put(5625,1962){\makebox(0,0)[lb]{{\SetFigFont{7}{8.4}{\rmdefault}{\mddefault}{\updefault}$\hat{\overline{W}}_{3,1}$}}}
\end{picture}
}\end{center}
\caption{$4$ node network with sum capacity not smaller than the sum capacity of the $4$ user full duplex network}
\label{fig:Kusertrans}
\end{figure}
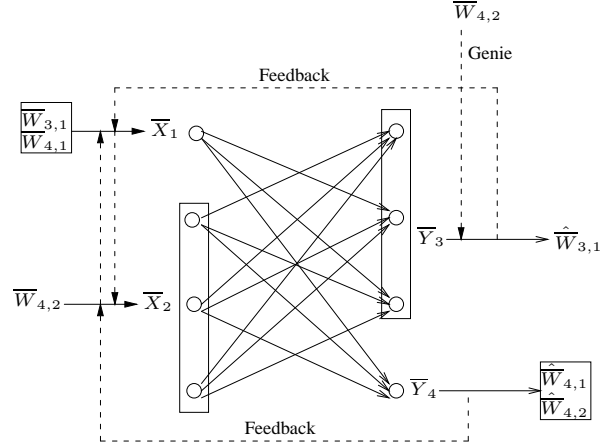

We now proceed to prove the outerbound of Theorem \ref{thm:Kuserfd}.

\subsection{Proof of Outer-bound of Theorem \ref{thm:Kuserfd}}
\label{subsec:fdouterbound}
Note that the outerbound of Theorem \ref{thm:Kuserfd} is equivalent to the following statement
$$d_{fd} \leq \frac{K(K-1)}{2K-3}$$ 
where $d_{fd}$ represents the number of degrees of freedom of the $K$ user full duplex network.

If $\mathcal{D}^{[K]}_{fd}$ is the degrees of freedom region of the $K$ user full duplex network, we show that 
$$ \sum_{j\in \{1,2,3, \ldots p-1,p+1, \ldots K\} } d_{j,p} + \sum_{i\in \{1,2,3, \ldots q-1,q+1, \ldots K\}} d_{q,i}  -d_{q,p}  \leq 1, \forall p \neq q$$
for all $\big( d_{i,j}\big) \in \mathcal{D}_{fd}^{[K]}  $.
Summing inequalities of the above form over all $(p,q), p \neq q$ gives the desired outerbound. 
It is enough to show the inequality for $p=1$ and $q=K$. The inequality extends to all other values of $(p,q)$ by symmetry. Therefore, we intend to show 
$$ \sum_{j\in \{2,3, \ldots K\} } d_{j,1} + \sum_{i\in \{1,2,3, \ldots K-1\}} d_{i,K} - d_{K,1}  \leq 1.$$
The above inequality is shown for the equivalent network described in Lemma \ref{lemma:Kusereq}. 
To show the above inequality, we first set $\widetilde{W}_{i,j} = \phi, (i-K)(j-1) \neq 0$.  With these messages set to null, there are no messages intended for destination node $1$ and therefore, it can only help the capacity of the network through feedback of the received symbol to node $1$. Therefore, we can delete side information of messages $\widetilde{W}_{j,1}, j=2,3 \ldots K$ at destination $1$ without affecting the converse argument. We now allow destination nodes $1 \ldots \ldots K-1$ to co-operate with each other, so that they form a single multi-antenna destination node.  Similarly, we allow source nodes $2 \ldots K$ to cooperate with each other so that they form a single multi-antenna source node. Note that allowing for co-operation does not reduce capacity and therefore can be used in the converse argument. The network is thus transformed to a $4$ node $X$ network (Figure \ref{fig:Kusertrans}) with input-output relations described by 
\begin{eqnarray}\label{eqn:KuserXtrans}
\overline{Y}_i(n)&=&\sum_{j=1}^2{\overline{H}}_{i,j}(n)\overline{X}_j(n)+\overline{Z}_i(n), ~~~~~i\in\{3,4\}
\end{eqnarray}
where 
\begin{eqnarray}
\overline{Y}_2(n) & = & \left[\widetilde{Y}_2(n)~~\cdots ~~\widetilde{Y}_{K-1}(n)\right]^T\\
\overline{Y}_4(n) & = & \widetilde{Y}_{K}(n)
\end{eqnarray}
Nodes $2$ and $3$ act as multiple antenna nodes, each with $K-1$ antennas. $\overline{X}_i(n), \overline{Z}_j(n)$ are also defined in a corresponding manner for $i=1,2, j=3,4$. The definition of the channel coefficients $\overline{H}_{i,j}(n)$ is clear from equations (\ref{eqn:KuserXtrans}) and (\ref{eq:def_fd}), and from Figures \ref{fig:Kuserfd}, \ref{fig:Kusertrans} and \ref{fig:4userX}. Note that $\overline{H}_{4,1} = H_{K,1} \neq 0$ since we have $p \neq q$.
The messages in this $4$ node $X$ network are defined as follows
\begin{eqnarray}
\overline{W}_{3,1}&=& \left[ \widetilde{W}_{2,1} ~~\widetilde{W}_{3,1} ~~\cdots~~\widetilde{W}_{K-1,1}\right]\\
\overline{W}_{3,2}&=& \phi\\
\overline{W}_{4,1} & = &  \widetilde{W}_{K,1}\\
\overline{W}_{4,2} & = & \left[\widetilde{W}_{K,2} ~~ \widetilde{W}_{q,3} ~~ \cdots ~~ \widetilde{W}_{K,K-1}\right]
\end{eqnarray}

The encoding and decoding functions, for codewords of length $K$ over this $4$ node $X$ network are defined as
\begin{eqnarray*}
\overline{X}_1(n) &=& \overline{f}_{1,n}(\overline{W}_{3,1},\overline{W}_{4,1},\overline{Y}_3^{n-1}, \overline{Y}_4^{n-1})\\
\overline{X}_2(n) &=& \overline{f}_{2,n}(\overline{W}_{4,2},\overline{Y}_3^{n-1}, \overline{Y}_4^{n-1})\\
\hat{\overline{W}}_{4,i} &=& \overline{g}_{4,i}(\overline{Y}_4^{N}), i=1,2\\
\hat{\overline{W}}_{3,1} &=& \overline{g}_{3,1}(\overline{Y}_3^{N}, \overline{W}_{4,2})
\end{eqnarray*}

We allow multi-antenna node $3$ to have apriori knowledge of message $\overline{W}_{3,1}$ through a genie. We also allow perfect feedback from destination nodes $3,4$ to source nodes $1,2$.  Note that the side information through feedback and genie in the $4$ user network constructed is stronger than the information at the corresponding nodes in the original $K$ user network. Since we are only providing an outerbound on the degrees of freedom region, the argument is not affected. Now, over this network, we claim that Lemma \ref{lemma:pq} holds. The $4$ node $X$ network differs from the network of Lemma \ref{lemma:pq} in two aspects: 
\begin{enumerate}
\item In the $X$ network constructed in this section, node $3$ has information of message $W_{4,2}$ apriori. In the network of Lemma \ref{lemma:pq}, node $3$ does not have this side information. 
\item The network constructed here is not fully connected since certain channel co-efficients are equal to zero. However $\overline{H}_{4,1}$ is non-zero.
\end{enumerate}
1) does not affect the converse argument because the proof of Lemma \ref{lemma:pq} begins with the genie providing information of $W_{4,2}$ (and $W_{4,1}$) to node $3$ (see figure \ref{fig:4userXconverse}). 2) does not affect the converse argument because, as mentioned towards the end of section \ref{subsec:main}, the proof of Lemma \ref{lemma:pq} holds as long as $ \overline{H}_{4,1} \neq 0$. Therefore, the bound of Lemma \ref{lemma:pq} holds for the $4$ user network in consideration here i.e. the network defined by equations (\ref{eqn:KuserXtrans}) and we can write
$$ \overline{d}_{3,1} + \overline{d}_{4,1} + \overline{d}_{4,2} \leq 1$$
$$ \Rightarrow \sum_{i=2}^{K}d_{i,1}+\sum_{j=1}^{K-1}d_{K,j}-d_{K,1} \leq 1$$
where $\overline{d}_{i,j}$ represents the number of degrees of freedom corresponding to message $W_{i,j}$. The desired result follows from the final equation above.
\QED

\subsection{Inner-bound of Theorem \ref{thm:Kuserfd}}
\label{subsec:fdinnerbound}
\begin{figure}[!tbp]
\begin{center}\setlength{\unitlength}{0.00050000in}
\begingroup\makeatletter\ifx\SetFigFont\undefined%
\gdef\SetFigFont#1#2#3#4#5{%
  \reset@font\fontsize{#1}{#2pt}%
  \fontfamily{#3}\fontseries{#4}\fontshape{#5}%
  \selectfont}%
\fi\endgroup%
{\renewcommand{\dashlinestretch}{30}
\begin{picture}(5025,2931)(0,-10)
\put(1575,2833){\ellipse{150}{150}}
\put(1575,1033){\ellipse{150}{150}}
\put(1575,133){\ellipse{150}{150}}
\put(3675,2833){\ellipse{150}{150}}
\put(3675,1933){\ellipse{150}{150}}
\put(3697,1011){\ellipse{150}{150}}
\put(3675,133){\ellipse{150}{150}}
\put(1575,1933){\ellipse{150}{150}}
\path(600,2833)(1125,2833)
\blacken\path(1005.000,2803.000)(1125.000,2833.000)(1005.000,2863.000)(1005.000,2803.000)
\path(600,1933)(1125,1933)
\blacken\path(1005.000,1903.000)(1125.000,1933.000)(1005.000,1963.000)(1005.000,1903.000)
\path(600,1033)(1125,1033)
\blacken\path(1005.000,1003.000)(1125.000,1033.000)(1005.000,1063.000)(1005.000,1003.000)
\path(600,133)(1125,133)
\blacken\path(1005.000,103.000)(1125.000,133.000)(1005.000,163.000)(1005.000,103.000)
\path(4200,2833)(4950,2833)
\path(4830.000,2803.000)(4950.000,2833.000)(4830.000,2863.000)
\path(4200,1033)(4950,1033)
\path(4830.000,1003.000)(4950.000,1033.000)(4830.000,1063.000)
\path(4200,133)(4950,133)
\path(4830.000,103.000)(4950.000,133.000)(4830.000,163.000)
\path(1650,2833)(3600,2008)
\path(3477.795,2027.128)(3600.000,2008.000)(3501.173,2082.386)
\path(1650,1933)(3600,1033)
\path(3478.473,1056.048)(3600.000,1033.000)(3503.617,1110.526)
\path(1650,1033)(3600,1933)
\path(3503.617,1855.474)(3600.000,1933.000)(3478.473,1909.952)
\path(1650,2833)(3600,1108)
\path(3490.243,1165.039)(3600.000,1108.000)(3529.998,1209.979)
\path(1650,133)(3600,958)
\path(3501.173,883.614)(3600.000,958.000)(3477.795,938.872)
\path(1650,133)(3600,1858)
\path(3529.998,1756.021)(3600.000,1858.000)(3490.243,1800.961)
\path(1650,1933)(3600,2833)
\path(3503.617,2755.474)(3600.000,2833.000)(3478.473,2809.952)
\path(1650,1033)(3600,2758)
\path(3529.998,2656.021)(3600.000,2758.000)(3490.243,2700.961)
\path(1650,1033)(3600,133)
\path(3478.473,156.048)(3600.000,133.000)(3503.617,210.526)
\path(1650,1933)(3600,208)
\path(3490.243,265.039)(3600.000,208.000)(3529.998,309.979)
\path(1650,2833)(3675,208)
\path(3577.950,284.690)(3675.000,208.000)(3625.457,321.338)
\path(1650,133)(3675,2758)
\path(3625.457,2644.662)(3675.000,2758.000)(3577.950,2681.310)
\path(4200,1933)(4950,1933)
\path(4830.000,1903.000)(4950.000,1933.000)(4830.000,1963.000)
\put(1200,2758){\makebox(0,0)[lb]{{\SetFigFont{7}{8.4}{\rmdefault}{\mddefault}{\updefault}$\widetilde{X}_1$}}}
\put(1200,1858){\makebox(0,0)[lb]{{\SetFigFont{7}{8.4}{\rmdefault}{\mddefault}{\updefault}$\widetilde{X}_2$}}}
\put(1200,958){\makebox(0,0)[lb]{{\SetFigFont{7}{8.4}{\rmdefault}{\mddefault}{\updefault}$\widetilde{X}_3$}}}
\put(3825,2758){\makebox(0,0)[lb]{{\SetFigFont{7}{8.4}{\rmdefault}{\mddefault}{\updefault}$\widetilde{Y}_1$}}}
\put(3825,1858){\makebox(0,0)[lb]{{\SetFigFont{7}{8.4}{\rmdefault}{\mddefault}{\updefault}$\widetilde{Y}_2$}}}
\put(3825,958){\makebox(0,0)[lb]{{\SetFigFont{7}{8.4}{\rmdefault}{\mddefault}{\updefault}$\widetilde{Y}_3$}}}
\put(3825,58){\makebox(0,0)[lb]{{\SetFigFont{7}{8.4}{\rmdefault}{\mddefault}{\updefault}$\widetilde{Y}_4$}}}
\put(5025,58){\makebox(0,0)[lb]{{\SetFigFont{7}{8.4}{\rmdefault}{\mddefault}{\updefault}$\hat{\widetilde{W}}_{4,3}$}}}
\put(5025,958){\makebox(0,0)[lb]{{\SetFigFont{7}{8.4}{\rmdefault}{\mddefault}{\updefault}$\hat{\widetilde{W}}_{3,2}$}}}
\put(5025,1858){\makebox(0,0)[lb]{{\SetFigFont{7}{8.4}{\rmdefault}{\mddefault}{\updefault}$\hat{\widetilde{W}}_{2,1}$}}}
\put(0,1933){\makebox(0,0)[lb]{{\SetFigFont{7}{8.4}{\rmdefault}{\mddefault}{\updefault}$\widetilde{W}_{3,2}$}}}
\put(0,58){\makebox(0,0)[lb]{{\SetFigFont{7}{8.4}{\rmdefault}{\mddefault}{\updefault}$\widetilde{W}_{1,4}$}}}
\put(0,958){\makebox(0,0)[lb]{{\SetFigFont{7}{8.4}{\rmdefault}{\mddefault}{\updefault}$\widetilde{W}_{4,3}$}}}
\put(75,2758){\makebox(0,0)[lb]{{\SetFigFont{7}{8.4}{\rmdefault}{\mddefault}{\updefault}$\widetilde{W}_{2,1}$}}}
\put(5025,2758){\makebox(0,0)[lb]{{\SetFigFont{7}{8.4}{\rmdefault}{\mddefault}{\updefault}$\hat{\widetilde{W}}_{1,4}$}}}
\put(1200,58){\makebox(0,0)[lb]{{\SetFigFont{7}{8.4}{\rmdefault}{\mddefault}{\updefault}$\widetilde{X}_4$}}}
\end{picture}
}\end{center}
\caption{Innerbound for the $4$ user full duplex channel - conversion to interference channel with $4$ messages}
\label{fig:Kuserachieve}
\end{figure}
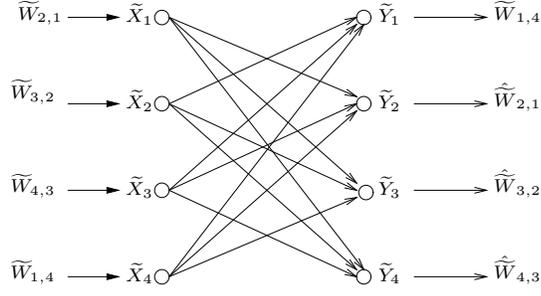

Note that the innerbound of theorem \ref{thm:Kuserfd} is equivalent to 
$$d_{fd} \geq \frac{K(K-1)}{2K-2} = \frac{K}{2}$$ 
where $d_{fd}$ represents the number of degrees of freedom of the $K$ user full duplex network.
The innerbound is shown using an achievable scheme based on interference alignment over multiple symbol extension of the channel, much like \cite{cadambe_jafar:dofx}. 
The argument is placed in appendix \ref{app:Kuserfd_ach}. However, here we consider the special case of the $K$ user full duplex channel where the reciprocity constraint on the channel is relaxed, i.e. the case where $H_{i,j}$ is independent of $H_{j,i}$. In this case, the innerbound can be shown using the fact that the $K$ user interference channel has $K/2$ degrees of freedom. To see this, over the equivalent $K$ user full duplex channel of Lemma \ref{lemma:Kusereq} we set $\widetilde{W}_{i,j} = \phi, \forall (i-j) \equiv l (\mbox{mod } K), \forall l \neq 1$. In other words, the only messages that are not set to null are $\widetilde{W}_{1,2}, \widetilde{W}_{2,3}, \ldots \widetilde{W}_{K-1,K}, \widetilde{W}_{K,1}$ (See Figure \ref{fig:Kuserachieve}). Note that with these messages, the channel is an interference channel with certain \emph{interfering} channels ($\widetilde{H}_{i,i})$ set to zero. Since interference can only hurt the capacity of a user, setting certain interference channels to zero can only increase the capacity. Therefore, the equivalent $K$ user full duplex network performs atleast as well as the $K$ user interference channel and therefore has at least $K/2$ degrees of freedom. Note that this argument does not hold if $\widetilde{H}_{i,j} = \widetilde{H}_{j,i}$ since the achievable scheme over the interference channel in \cite{cadambe_jafar:Kuserint} requires each channel co-efficient to have a continuous probability distribution, given \emph{all} other channel co-efficients. (If the channels satisfy the reciprocity property, then a particular channel co-efficient is deterministic if all other channel co-efficients are given). However, we show that the innerbound is still valid and the achievable scheme is presented in Appendix \ref{app:Kuserfd_ach}.

\section{Conclusion}
We characterize the capacity of a fully connected network with $S$ source nodes, $R$ relays and $D$ destination nodes with full duplex operation and feedback. We also provide bounds on capacity approximations within $o(\log(\mbox{SNR}))$ of the $K$ user fully connected network in which there is a message from every node to every other node. The lower and upper bound provided are tight if $K$ is large. Apart from the small gap between the bounds of the $K$ user fully connected network, this work effectively solves the degrees of freedom problem for a fairly large class of wireless networks with time-varying/frequency-selective channel gains. 

A major implication of our result is that in fully connected networks whose source nodes are different from destination nodes, the techniques of relays, perfect feedback to source nodes, noisy co-operation and full duplex operation can only increase the capacity upto $o(\log(\mbox{SNR}))$ bits. In such networks, the search for improvements in capacity of the order of $\log(\mbox{SNR})$ ends in interference alignment. In other words, the mentioned techniques (relays etc.) cannot improve the \emph{degrees of freedom} of wireless networks except under certain special circumstances listed below. 
\begin{enumerate}
\item Relays can increase the degrees of freedom of networks that are not fully connected
\item Feedback can increase the degrees of freedom of a network, if it is provided to a decoding node. However, feedback to source/relay nodes does not increase the degrees of freedom of wireless networks.
\item Full duplex operation can increase the degrees of freedom of wireless networks if certain source nodes also behave as destination nodes i.e. the set of source nodes and the set of destination nodes are \emph{not} disjoint.
\item Cognitive co-operation can increase the degrees of freedom of wireless networks if we do not account for the cost of a cognition (for example, the cost of acquiring a message) at a node.
\end{enumerate}
It must be observed that our result does not necessarily discourage the use of relays, feedback, co-operation and full duplex operation in real communication scenarios. In fact, not all wireless networks are fully connected under realistic transmit signal powers. Furthermore, in several communication networks, source nodes act as destination nodes as well. Importantly, the result of this paper provides an insight into the type of networks which are most likely to benefit from these techniques, especially at high SNR. Furthermore, our result does not preclude huge benefits in terms of capacity at low or mid-range SNR from the mentioned techniques which continues to be an important area of research in wireless communications. Another important limitation of our results is the assumption of time-varying and/or frequency selective channel gains for the achievability schemes based on interference alignment. However, the outerbounds of Theorem \ref{thm:main} and Theorem \ref{thm:Kuserfd} are fairly general and hold for all fully connected networks whether the channel coefficients are time-varying or constant.

\appendices
\section{Proof of Lemma \ref{lemma:U}}
\label{proof:lemmaU}
We use induction principle to prove Lemma \ref{lemma:U}. First, we argue that $S_1(1), S_2(1), S_3(1)$ are true. 
At the first channel use ($n=1$), because there are no prior received signals, the destination nodes $3,4$ have no useful information to transmit. Therefore, without loss of generality $\overline{X}_3(1), \overline{X}_4(1) = 0$. 
\begin{eqnarray}
\overline{X}_1(1) &=& \overline{f}_{1,1}\left(\overline{W}_{3,1},\overline{W}_{4,1}\right)\label{eq:ind11}\\
\overline{X}_2(1) &=& \overline{f}_{2,1}\left(\overline{W}_{4,2}\right)\label{eq:ind12}\\
\overline{Y}_i(1)&=& \overline{U}_i(1) + \overline{H}_{i,2}(1)\overline{X}_2(1), ~~i\in\{1,2,3,4\}\label{eq:ind13}
\end{eqnarray}
Equations (\ref{eq:ind11}), (\ref{eq:ind12}) imply that the statements $S_1(1), S_2(1)$ are true. (\ref{eq:ind12}), together with (\ref{eq:ind13}) implies that $S_3(1)$ is true as well. Thus, the result of Lemma \ref{lemma:U} holds for $n=1$. For $n=2$,
\begin{eqnarray}
\overline{X}_1(2) &=& \overline{f}_{1,2}\left(\overline{W}_{3,1},\overline{W}_{4,1}, \overline{Y}_1(1), \overline{Y}_2(1), \overline{Y}_3(1), \overline{Y}_4(1)\right)\label{eq:ind21}\\
\overline{X}_2(2) &=& \overline{f}_{2,2}\left(\overline{W}_{4,2}, \overline{Y}_1(1), \overline{Y}_2(1), \overline{Y}_3(1), \overline{Y}_4(1)\right)\label{eq:ind22a}\\
\overline{X}_j(2) &=&  \overline{f}_{j,2}\left(\overline{Y}_j(1)\right), ~~j\in\{3,4\}\label{eq:ind22}\\
\overline{Y}_i(2)&=& \overline{U}_i(2) + \sum_{j=2}^4\overline{H}_{i,j}(2)\overline{X}_j(2), ~~i\in\{1,2,3,4\}\label{eq:ind23}
\end{eqnarray}
Equation (\ref{eq:ind21}) and $S_1(1)$ imply that $S_1(2)$ is true. (\ref{eq:ind22a}), (\ref{eq:ind22}) and $S_2(1)$ imply that $S_2(2)$ is true. $S_1(2),S_2(2)$ and (\ref{eq:ind23}) together imply that $S_3(2)$ is true as well.
Thus, $S_i(n),i=1,2,3$ hold true for $n=2$. Following the induction argument, suppose that $S_i(n),i=1,2,3$ are true for $n=k-1$. Then,
\begin{eqnarray*}
\overline{X}_1(k) &=& \overline{f}_{1,k}\left(\overline{W}_{3,1},\overline{W}_{4,1},\overline{Y}_1^{k-1}, \overline{Y}_2^{k-1}, \overline{Y}_3^{k-1}, \overline{Y}_4^{k-1}\right)\\
\overline{X}_2(k) &=& \overline{f}_{2,k}\left(\overline{W}_{4,2}, \overline{Y}_1^{k-1}, \overline{Y}_2^{k-1}, \overline{Y}_3^{k-1}, \overline{Y}_4^{k-1}\right)\\
\overline{X}_j(k) &=&  \overline{f}_{j,2}\left(\overline{Y}_j^{k-1}\right), ~~j\in\{3,4\}\\
\overline{Y}_i(k)&=& \overline{U}_i(k) + \sum_{j=2}^4\overline{H}_{i,j}(k)\overline{X}_j(k), ~~i\in\{1,2,3,4\}
\end{eqnarray*}
The above equations, along with $S_1(k-1), S_2(k-1), S_3(k-1)$ imply that $S_1(k), S_2(k), S_3(k)$ are true. The inductive assumption implies that $S_i(n),i=1,2,3$ hold for all values of $n$.

\section{Proof of innerbound of Theorem \ref{thm:Kuserfd} : Achievable scheme}
\label{app:Kuserfd_ach}

The achievability proof is based on interference alignment over the channel described in lemma \ref{lemma:Kusereq} (Figure \ref{fig:Kuserfd}(b)). Since many of the details are identical to \cite{cadambe_jafar:dofx}, we focus here on the unique aspects of this proof.

Let $\Gamma = (K-1)(K-2)$. We show that $K(K-1) n ^\Gamma$ degrees of freedom are achievable over a $$ \mu_n  = (K-1) \left( (n+1)^\Gamma + n^\Gamma\right)$$ symbol extension of the channel for any $n \in \mathbb{N}$ thus implying the desired result. 
 Over the extended channel, the scheme achieves $n^{\Gamma}$ degrees of freedom for each of the $K(K-1)$ messages $\widetilde{W}_{i,j}, j\neq i$.
The signal vector in the extended channel at the $j^{th}$ user's receiver can be expressed as 
$$ \xY_{j}(\kappa) = \displaystyle\sum_{i=1}^{M} \xH_{j,i}(\kappa) \mathbf{X}_{i}(\kappa) + \xZ_{j}$$ 
where $\mathbf{X}_{i}$ is a $\mu_n \times 1$ column vector representing the $\mu_n$ symbol extension of the transmitted symbol $X_{i}$, i.e 
$$\xX_{i}(\kappa) \define \left[ \begin{array}{c} \widetilde{X}_{i}(\mu_n\kappa+1)\\ \widetilde{X}_{i}(\mu_n \kappa + 2)\\ \vdots \\ \widetilde{X}_{i}(\mu_n (\kappa+1)) \end{array}\right]$$
Similarly $\xY_{i}$ and $\xZ_{i}$ represent $\mu_n$ symbol extensions of the $\widetilde{Y}_{i}$ and $\widetilde{Z}_{i}$ respectively. 
$\xH_{i,j}$ is a diagonal $\mu_n\times \mu_n$ matrix representing the $\mu_n$ symbol extension of the channel.
Similar to the interference alignment based achievable schemes of the interference and $X$ channels, the message $\widetilde{W}_{i,j}$ is encoded at transmitter $j$ as $n^\Gamma$ independent streams so that $\xX_j$ is 
$$\xX_{j}(\kappa) = \displaystyle\sum_{i=\{1,2\ldots K\} - \{j\} }\displaystyle\sum_{m=1}^{(n+1)^\Gamma} x_{i,j}^{[m]}(\kappa) \mathbf{v}^{[m]}_{i,j}( \kappa ) = \displaystyle\sum_{i=\{1,2, \ldots K\} - \{j\}}\xV_{i,j}(\kappa) \mathbf{x}_{i,j}(\kappa)$$
The received signal at the $k^{th}$ receiver can then be written as
$$ \xY_{k}(\kappa) = \displaystyle\sum_{i=1}^{M}\xH_{k,i}(\kappa) \big( \displaystyle\sum_{j=1}^{N}\xV_{j,i}(\kappa) \mathbf{x}_{j,i}(\kappa) \big) + \xZ_{k}(\kappa)$$

We now need to ensure that at receiver $j$, the $(K-1)(K-2)$ interfering spaces $\xV_{k,i},k \neq i, k \neq j, i \neq j$ lie in a $(K-1)(n+1)^\Gamma$ dimensional space so that $(K-1) n^\Gamma$ desired spaces $\xV_{j,i}, i \in \{1,2 \ldots \} - \{j\}$ can be decoded free of interference from a $\mu_n$ dimensional space. To do this, we first set 
$$\xV_{j,i} = \xV_{j}, \forall i \neq j $$ 
Then, we design $\xV_{j}, j=1,2 \ldots K$ so that they satisfy the interference alignment equations below.
\begin{equation} \label{eqn:intalign} \xH_{i,j} \xV_{k} \prec \xI_{k}, \forall \{(i,j,k): i \neq k, k \neq j, j \neq i\} \end{equation} 
such that $\mbox{rank}(\xI_{k}) = (n+1)^\Gamma$ where $\mathbf{P} \prec \mathbf{Q}$ implies that the span of the column vectors of $\mathbf{P}$ lies in the vector space spanned by the column vectors of $\mathbf{Q}$.. Note that for a fixed $k$, there are $\Gamma=(K-1)(K-2)$ relations of the above form.
We first generate $\mu_n \times 1$ column vectors $\xw_{k}, k=1,2\ldots K$ so that all the entries of $\xw_{k}$ are drawn from any continuous distribution independently from each other and independently from all other entries in $\xw_{l}, l \neq k$. The rest of the proof is similar to the achievable scheme for the $X$ channel presented in \cite{cadambe_jafar:dofx}. It is easy to observe that the dimension of the interfering space at receiver $k$ space is equal to the dimension of the space spanned by all column vectors of matrices $\xI_j,j\neq k$ which is equal to $(K-1)(n+1)^\Gamma$. The only difference from the model in \cite{cadambe_jafar:dofx} is that here, we have $\xH_{i,j}= \xH_{j,i}$ whereas, in \cite{cadambe_jafar:dofx} the matrix $\xH_{j,i}$ is independent from $\xH_{i,j}$. However this difference does not affect the construction of vectors satisfying the desired  interference alignment relations (\ref{eqn:intalign}). The difference does not affect the argument that at any receiver, the signal space is linearly independent with the interference space since the argument only depends on $\xw_{k}$ being independent of $\xw_{l}$ for $l \neq k$. The only condition that needs to be verified is that all the desired streams of at receiver $k$ are linearly independent of each other. In other words, all that needs to be shown is that the column vectors of 
\begin{eqnarray*} \mathbf{D}_k &=& \left[ \xH_{k,1} \xV_{k,1} ~~ \xH_{k,2} \xV_{k,2}~~\ldots~~\xH_{k,k-1} \xV_{k,k-1}~~\xH_{k,k+1}\xV_{k,k+1}~~ \xH_{k,K} \xV_{k,K}\right] \\
&=& \left[ \xH_{k,1} \xV_{k} ~~ \xH_{k,2} \xV_{k}~~\ldots~~\xH_{k} \xV_{k}~~\xH_{k,k+1}\xV_{k}~~ \xH_{k,K} \xV_{k}\right] 
\end{eqnarray*}
are linearly independent.  The linear independence follows from the fact that the construction of $\xV_{k}$ satisfying the relations of (\ref{eqn:intalign}) is independent of both $\xH_{k,i}$ and $\xH_{i,k}$ for $i \neq k$. Again, the reader is referred to the achievable scheme in \cite{cadambe_jafar:dofx} for a formal proof of the same.
\bibliographystyle{ieeetr}
\bibliography{refs}

\end{document}